\documentclass[twoside,leqno,twocolumn]{article}
\usepackage[backend=biber,
            style=numeric-comp,
            natbib=true,
            maxbibnames=10,
            minalphanames=3,
            maxcitenames=2,
            sorting=ynt,
            sortcites=true]{biblatex}
\usepackage[letterpaper]{geometry}

\usepackage{ltexpprt}

\usepackage{macro_math}
\newcommand{\set}[1]{\{#1\}}

\renewcommand{\epsilon}{\varepsilon}

\usepackage{graphicx}
\usepackage{subfloat}
\usepackage{caption,subcaption}
\usepackage{algorithm,algorithmicx,algpseudocode}
\usepackage{color}
\usepackage{bold-extra}
\usepackage{enumitem}
\usepackage{tabularx}

\newcommand{\latent}{\delta_E}

\newcommand{\recover}{\delta_I}

\algnewcommand\algorithmicinput{\textbf{Input:}}
\algnewcommand\Input{\item[\algorithmicinput]}
\algnewcommand\algorithmicoutput{\textbf{Output:}}
\algnewcommand\Output{\item[\algorithmicoutput]}
\algnewcommand\algorithmicparam{\textbf{Parameters:}}
\algnewcommand\Param{\item[\algorithmicparam]}
\newcommand{\LPshort}{{\sc Frank-Wolfe-EC}}
\newcommand{\edgecenshort}{{\sc Top-k-EC}}
\newcommand{\LP}{{\sc Frank-Wolfe-EdgeCentrality}}
\newcommand{\LPTV}{{\sc Frank-Wolfe-TimeVarying}}
\newcommand{\LPTVshort}{{\sc Frank-Wolfe-TV}}

\newcommand{\edgecen}{{\sc Top-k-EdgeCentrality}}
\newcommand{\edgecenTV}{{\sc Top-k-TimeVarying}}

\makeatletter
\def\thanksnosymbol#1{\protected@xdef\@thanks{\@thanks
        \protect\footnotetext{#1}}}
\makeatother
\usepackage{enumitem}
\begin{document}

\newcommand\relatedversion{}
\renewcommand\relatedversion{\thanks{The full version of the paper can be accessed at \protect\url{}}} %

\title{Optimal Intervention on Weighted Networks via Edge Centrality}
\author{Dongyue Li\thanks{Northeastern University, Boston, MA. Email correspondence can be directed to all authors at $\langle${li.dongyu, t.eliassirad, ho.zhang$\rangle$@northeastern.edu}.}
\and Tina Eliassi-Rad\footnotemark[1]
\and Hongyang R. Zhang\footnotemark[1]
}
\date{}

\maketitle

 \fancyfoot[R]{\scriptsize{Copyright \textcopyright\ 2023 by SIAM\\
 Unauthorized reproduction of this article is prohibited}}

\pagenumbering{arabic}
\setcounter{page}{1} %
\begin{abstract}
Suppose there is a spreading process such as an infectious disease propagating on a graph. How would we reduce the number of affected nodes in the spreading process? This question appears in recent studies about implementing mobility interventions on mobility networks (Chang et al.\;(2021)). A practical algorithm to reduce infections on unweighted graphs is to remove edges with the highest edge centrality score (Tong et al.\;(2012)), which is the product of two adjacent nodes' eigenscores. However, mobility networks have weighted edges; Thus, an intervention measure would involve edge-weight reduction besides edge removal. Motivated by this example, we revisit the problem of minimizing top eigenvalue(s) on weighted graphs by decreasing edge weights up to a fixed budget. We observe that the edge centrality score of Tong et al.\;(2012) is equal to the gradient of the largest eigenvalue of $WW^{\top}$, where $W$ denotes the weight matrix of the graph. We then present generalized edge centrality scores as the gradient of the sum of the largest $r$ eigenvalues of $WW^{\top}$. With this generalization, we design an iterative algorithm to find the optimal edge-weight reduction to shrink the largest $r$ eigenvalues of $WW^{\top}$ under a given edge-weight reduction budget. We also extend our algorithm and its guarantee to time-varying graphs, whose weights evolve over time. We perform a detailed empirical study to validate our approach. Our algorithm significantly reduces the number of infections compared with existing methods on eleven weighted networks. Further, we illustrate several properties of our algorithm, including the benefit of choosing the rank $r$, fast convergence to global optimum, and an almost linear runtime per iteration.
\end{abstract}

\maketitle
\section{Introduction}

\begin{figure*}[!t]
    \hspace{0.06\textwidth}
  \begin{subfigure}[t]{0.48\textwidth}
    \centering
    \includegraphics[width=0.9\textwidth]{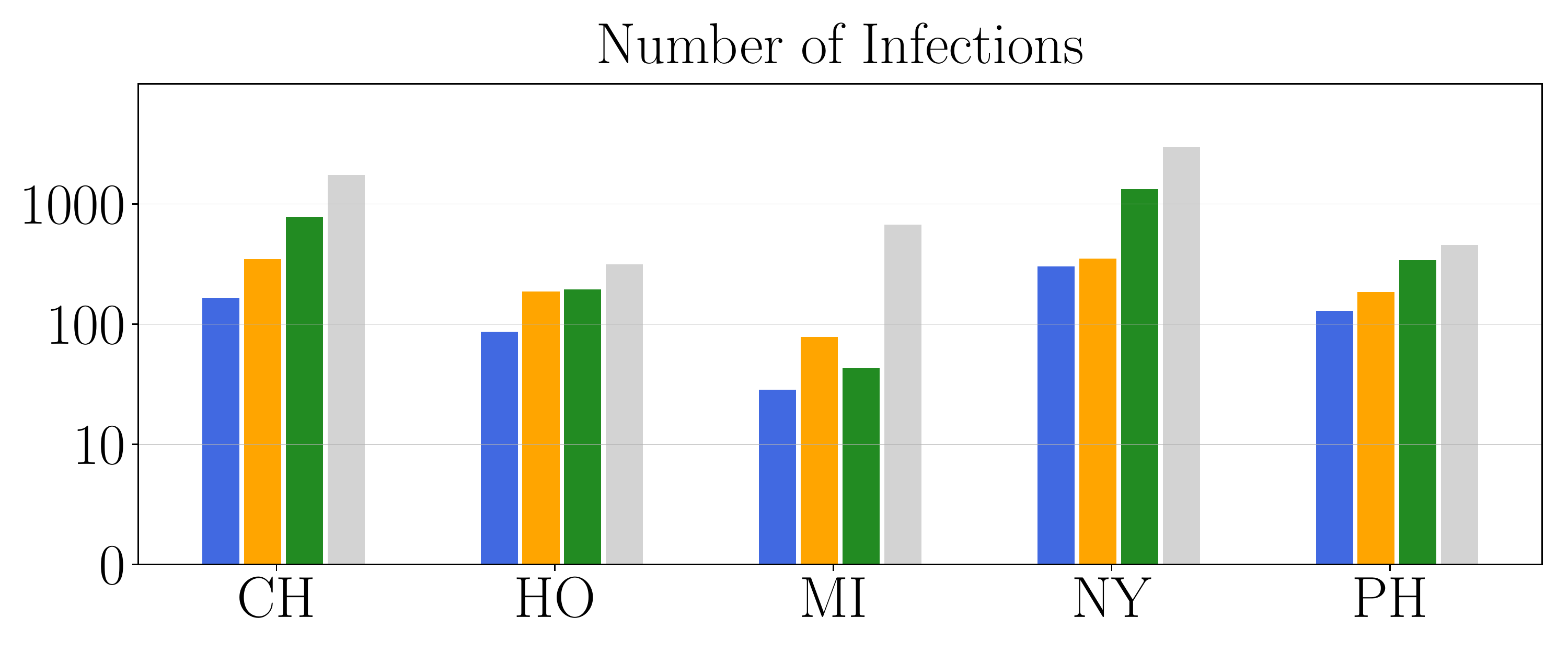}
  \end{subfigure}
  \begin{subfigure}[t]{0.32\textwidth}
    \centering
    \includegraphics[width=0.9\textwidth]{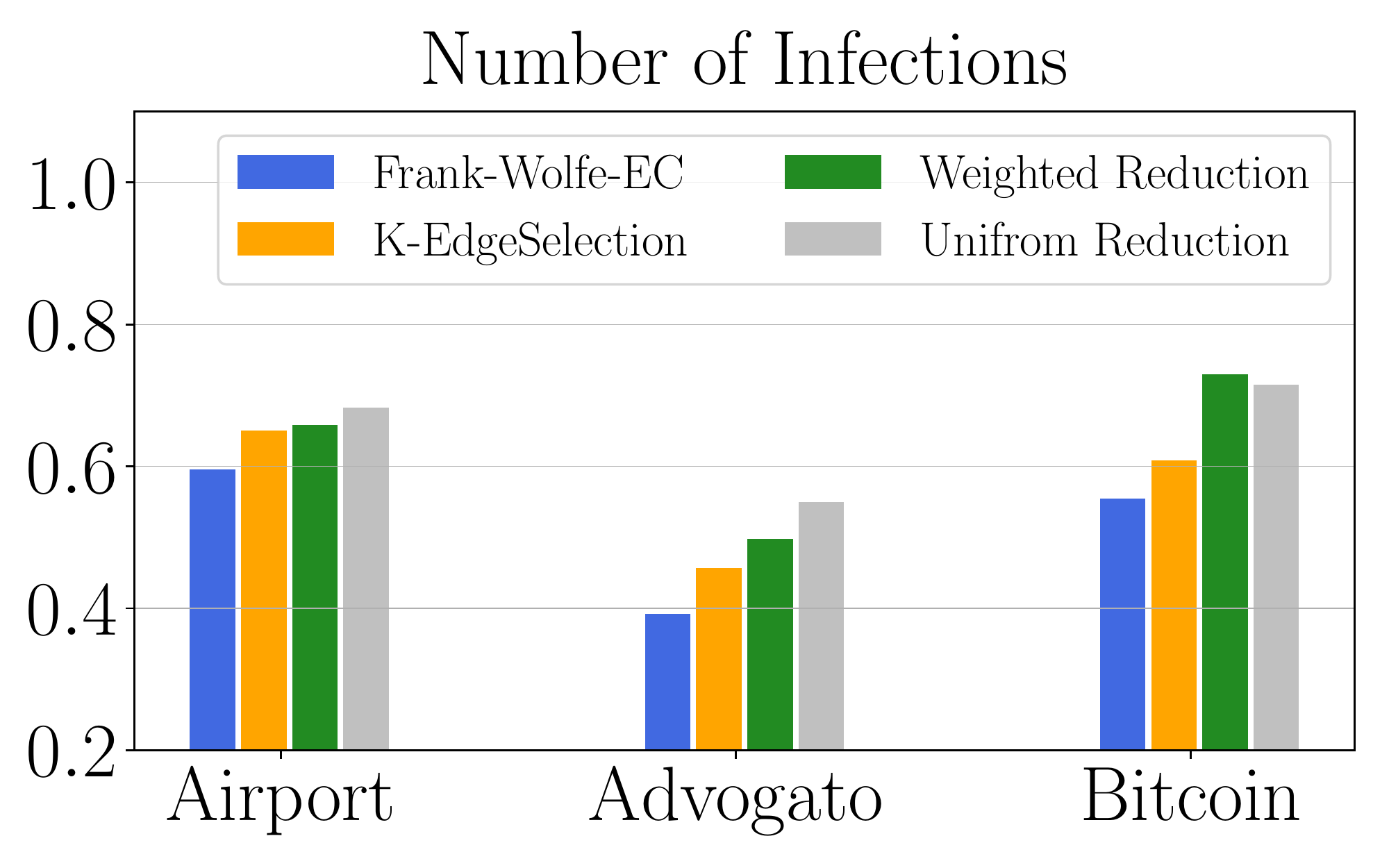}\label{fig_web_count}
  \end{subfigure}\vfill
  \hspace{0.06\textwidth}
  \begin{subfigure}[t]{0.48\textwidth}
    \centering
    \includegraphics[width=0.9\textwidth]{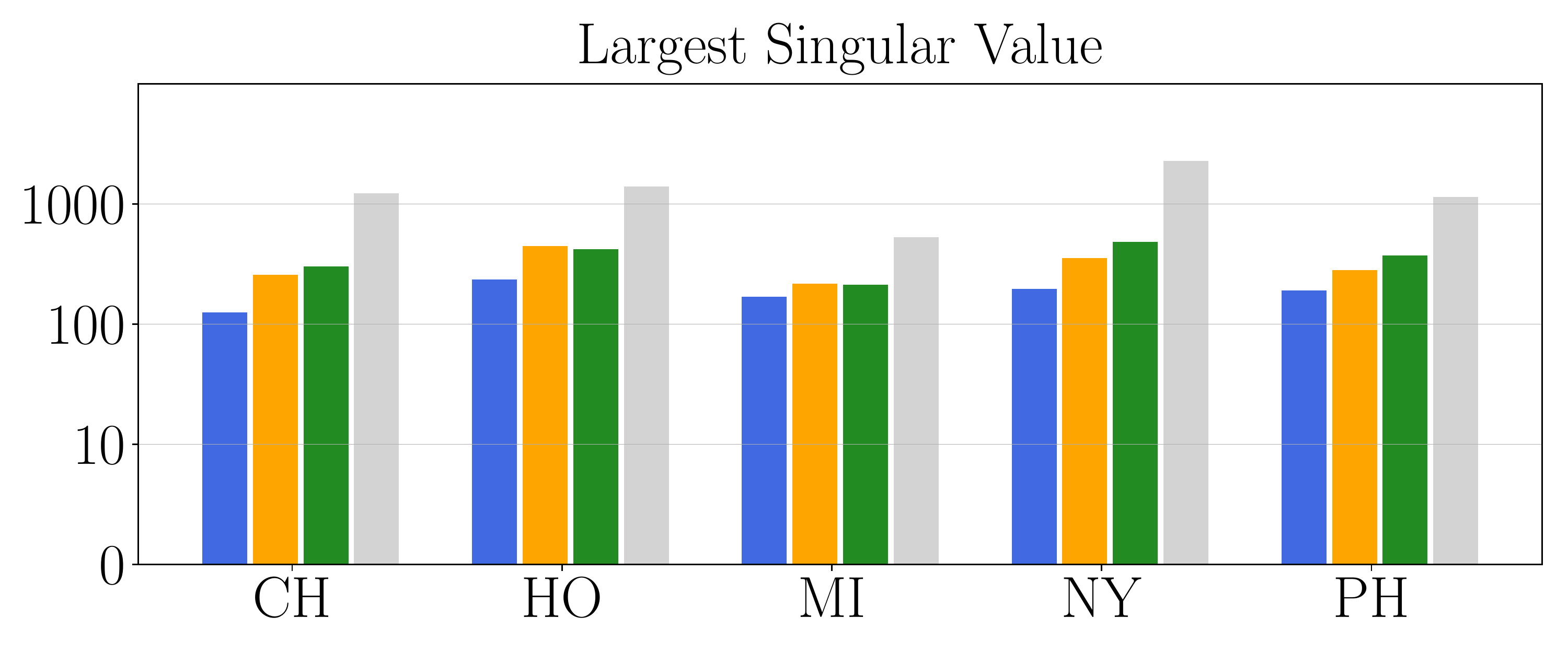}
  \end{subfigure}
  \begin{subfigure}[t]{0.32\textwidth}
    \centering
    \includegraphics[width=0.9\textwidth]{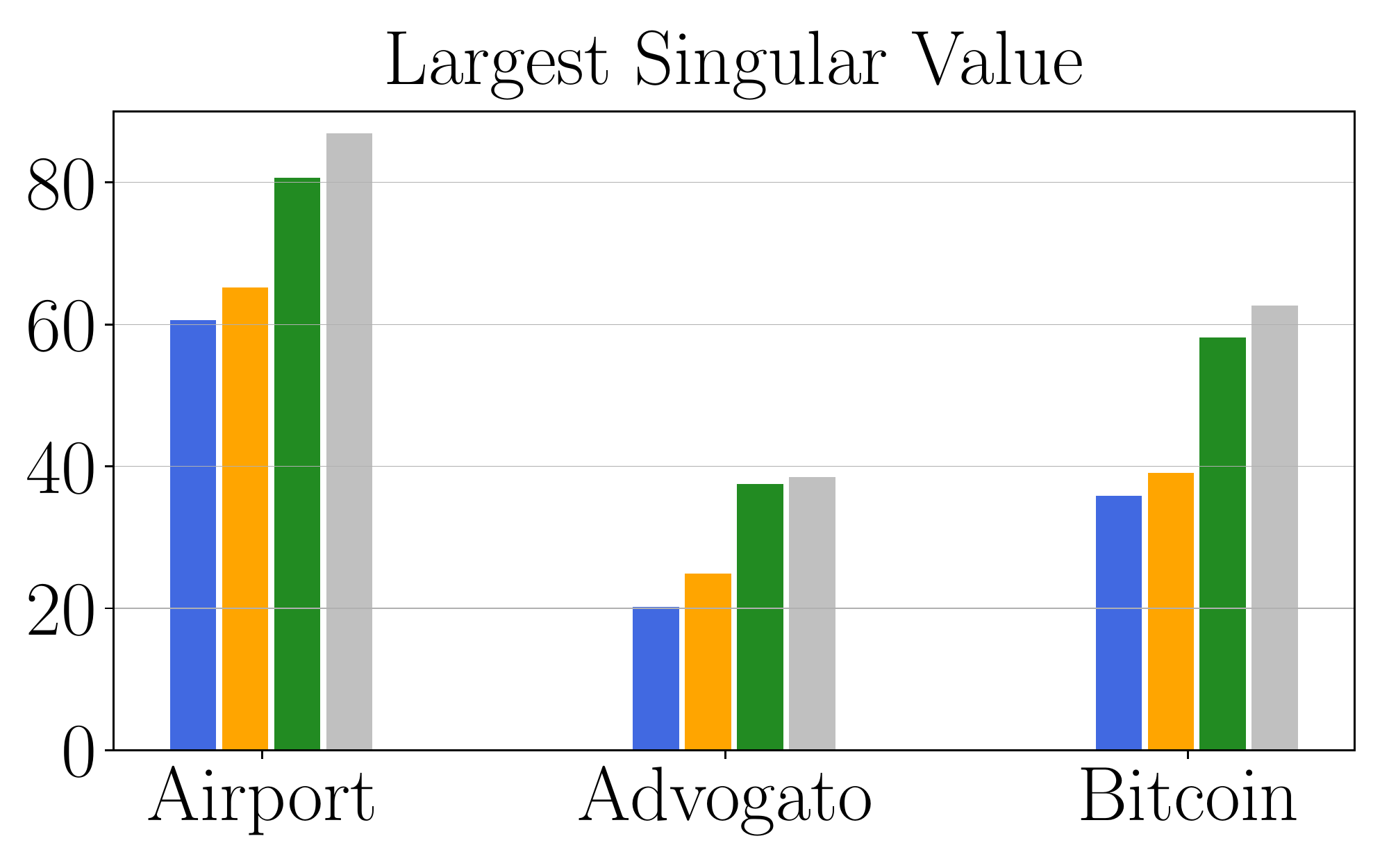}\label{fig_web_sin}
  \end{subfigure}
  \caption{Comparison of our algorithm (namely, Frank-Wolfe-EC) and several existing approaches, including K-EdgeSelection, Weighted Reduction, and Uniform Reduction (See Section \ref{sec_exp_setup} for a description of these approaches). 
  On the top panel, we report the number of infections (whose scale should be multiplied by $10^3$), averaging over fifty simulations.
  We observe that our approach can be used to reduce the number of infections and the largest singular value of the weight matrix of the diffusion process.
  At a high level, our approach works by:
  (i) Connecting the edge centrality score with the gradient of the sum of the top singular values; See Lemma \ref{prop_grad}.
  (ii) Showing that each iteration of the Frank-Wolfe algorithm can be solved efficiently with a greedy selection procedure; See Lemma \ref{thm_optimal_descent}.
  Moreover, we show that this approach applies to both static and time-varying graphs.
  }\label{fig_intro}
\end{figure*}

Suppose there is a spreading process, such as an epidemic propagating through a graph. Denote the graph as $G = (V, E)$. How would we reduce the number of affected nodes from $V$ during the spreading process? Many studies have considered this question in the network immunization literature \cite{chen2015node,chen2016eigen}, motivated by considerations for controlling the outcome of the diffusion process \cite{pastor2015review}.
A principal approach from the existing literature is to optimize spectral properties of $G$ with edge removal procedures.
For example, \citet{tong2012gelling} design algorithms to reduce the largest eigenvalue of $G$'s adjacency matrix by removing a budgeted number of edges.
\citet{le2015met} further study how to reduce the largest $r$ eigenvalues under a budget constraint of edge removals.
In this work, we revisit the spectral optimization approach on weighted graphs.
Let $W$ denote a non-negative weight matrix corresponding to the edge weights of $G$.
We consider edge-weight reduction with a budgeted amount of $B$ that will create the most drop in the largest $r$ eigenvalues of $WW^{\top}$.

For example, weighted graphs have appeared in recent studies about the pandemic.
\citet{chang2021mobility} study the counterfactual outcome of implementing edge-weight reduction strategies in mobility networks.
Reducing edge weights in mobility networks corresponds to restricting the mobility of population groups.

An effective algorithm for reducing the top singular values  of a graph is by removing edges with the highest centrality scores \cite{tong2012gelling}.
Let $\lambda_1(W)$ denote the largest singular value of $W$ (notice that the largest eigenvalue of $WW^{\top}$ is equal to the square of $\lambda_1(W)$).
Let $\vec u_1$ and $\vec v_1$ denote the left and right singular vectors corresponding to $\lambda_1(W)$, respectively.
The edge centrality score of an edge $(i, j)$ is equal to $\vec u_1(i)\cdot\vec v_1(j)$, where $\vec u_1(i)$ is the $i$-th entry of $\vec u_1$ and $\vec v_1(j)$ is the $j$-th entry of $\vec v_1$.
\citet{tong2012gelling} show that removing edges with the highest edge centrality scores effectively reduces $\lambda_1(W)$.
\citet{chen2018network} further quantifies the approximation ratio of this greedy algorithm using submodular optimization techniques (see also \citet{saha2015approximation}).
These works focus on the case of unweighted graphs, for which the spectral optimization problem given a budgeted amount of edge removals is NP-hard \cite{chen2016eigen}.
Notice that in the case of weighted graphs, the weight of an edge can be reduced by a fraction. 
\citet{yu2021potion} apply gradient-based optimization for targeted diffusion, which also applies to weighted graphs, with a stopping criterion until the gradient gets close to zero.

To motivate our approach, we begin by observing that the edge centrality score from the work of \citet{tong2012gelling} is equal to the gradient of the largest singular value of $W$ squared, up to a scaling of $2\lambda_1(W)$ (See Lemma \ref{prop_grad} for the full statement):
{\[ \frac{\partial\Big(\big(\lambda_1(W)\big)^2\Big)}{\partial W_{i, j}} = 2\lambda_1(W) \cdot \vec u_1(i) \cdot \vec v_1(j). \]}%
Notice that the above corresponds to the rank-$1$ SVD of $W$.
More generally, for any rank $r$, the gradient of the largest $r$ singular values can be efficiently computed via a rank-$r$ SVD of $W$.
Based on the connection between edge centrality and gradients, we minimize the largest $r$ eigenvalues of $WW^{\top}$ via the Frank-Wolfe algorithm, which involves direction finding and line search.
We show an efficient way to find the descent direction by reducing edges with the highest generalized edge centrality score (see Lemma \ref{thm_optimal_descent}).
We then recompute the eigenscores at each iteration, which is also related to the approach of \citet{le2015met}.
By comparison, our algorithm adapts to weighted graphs and is guaranteed to converge to the global optimum (see Theorem \ref{prop_continuous}).

With the connection between edge centrality and gradients, we extend our algorithm to time-varying networks, which include a sequence of graphs with evolving weight matrices.
We provide the generalized eigenscore for each edge of every graph in the sequence and design an algorithm for optimizing the largest $r$ eigenvalues of the product of all weight matrices in the sequence (cf. \citet[Sec. 4.2]{prakash2010virus}).

We evaluate our algorithms by simulating an epidemic model on eleven weighted graphs.
In the static case, our approach achieves, on average, $\mathbf{25.5\%}$ improvement over baselines during SEIR model simulations (cf. Section \ref{sec_prelim} for descriptions).
The largest singular value decreases by an average of $\mathbf{25.1\%}$ more than the baselines.
See Figure \ref{fig_intro} for an illustration.
Meanwhile, our approach is also effective for SIR and SIS models (see Appendix \ref{sec_epi_models}, where we describe both models).
Further, our algorithm reduces the number of infections by over $\mathbf{6.9\%}$ for several time-varying networks.

\smallskip
\noindent\textbf{Organization.} The rest of our paper is organized as follows.
In Section \ref{sec_prelim}, we formally define the spectral optimization problem on weighted graphs.
Then in Section \ref{sec_alg}, we develop two algorithms for this problem on static and time-varying networks.
We validate our approach with extensive experiments in Section \ref{sec_exp}.
Lastly, we discuss several related pieces of literature in Section \ref{sec_related} and questions for future work in Section \ref{sec_discuss}.

\section{Preliminaries}\label{sec_prelim}

\noindent\textbf{Problem setup.}
Given a spreading process on a network, we are interested in designing algorithms to reduce the number of affected nodes.
Let $\cG = (\cV, \cE)$ be a weighted and possibly directed graph.
Let $\cV$ be the set of vertices and $\cE$ be the set of edges.
We use $W$ to denote a non-negative weight matrix over the edges, with $W_{i,j}$ being its $(i,j)$-th entry.
Given an \emph{arbitrary} edge-weight reduction budget $B$,
how should we allocate the budget across the edges?

To answer this question, we consider an eigenvalue optimization approach that has been the basis of prior works for unweighted graphs \cite{chakrabarti2008epidemic,prakash2012threshold,tong2012gelling,chen2015node}.
The idea behind eigenvalue optimization approaches is to modify the weight matrix $W$ so that its largest eigenvalue is most reduced.
We extend the eigenvalue minimization approach to weighted networks as follows.
Let $M$ be an $n$ by $n$ matrix, where $n$ is the number of nodes in $\cV$.
Given a rank $r$, let $\lambda_k(M)$ be the $k$-th largest singular value of $M$.
We consider the following problem:
{\small
\begin{align}
     \min_{M}~~&\quad f(M) = \sum_{k=1}^r \big(\lambda_k(M)\big)^2 \label{eq_convex} \\
    \mbox{s.t.}              
                                ~~&\,\, \sum_{(i,j)\in \cE} \big(W_{i, j} - M_{i,j}\big) \le B \nonumber \\
                                \quad&\quad 0 \le M_{i, j} \le W_{i, j},\, \forall\, (i, j)\in \cE, \nonumber \\
                                \quad~~&\quad M_{i,j}=0, \quad\quad\quad\,\,\, \forall\, (i, j)\notin \cE. \nonumber
\end{align}}%
After solving the above problem, we get a reduced weight matrix $M$ as the solution of our intervention strategy.
Notice that we approach this problem from an optimization perspective. Questions including interpreting the solution would be interesting questions for future work.
As a remark, the square of $\lambda_k(M)$ equals the $k$-th largest eigenvalue of $MM^{\top}$.
Thus, the objective in equation \eqref{eq_convex}  includes the top-$r$ eigenvalues (see also \citet{le2015met}).
The reason is that the other top eigenvalues could still affect the spreading process in subgraphs of $G$ \cite{andersen2006local,gleich2012vertex,le2015met,yu2021potion}. %
In Figure \ref{fig:scale_singular_values}, we first illustrate that reducing the largest singular value of $G$ reduces the number of infections during simulated spreading processes.
In Section \ref{sec:scalability}, we further demonstrate that having the freedom to choose the rank $r$ helps reduce the number of infections.

\begin{figure}[!h]
    \centering
    \includegraphics[width=0.3234\textwidth]{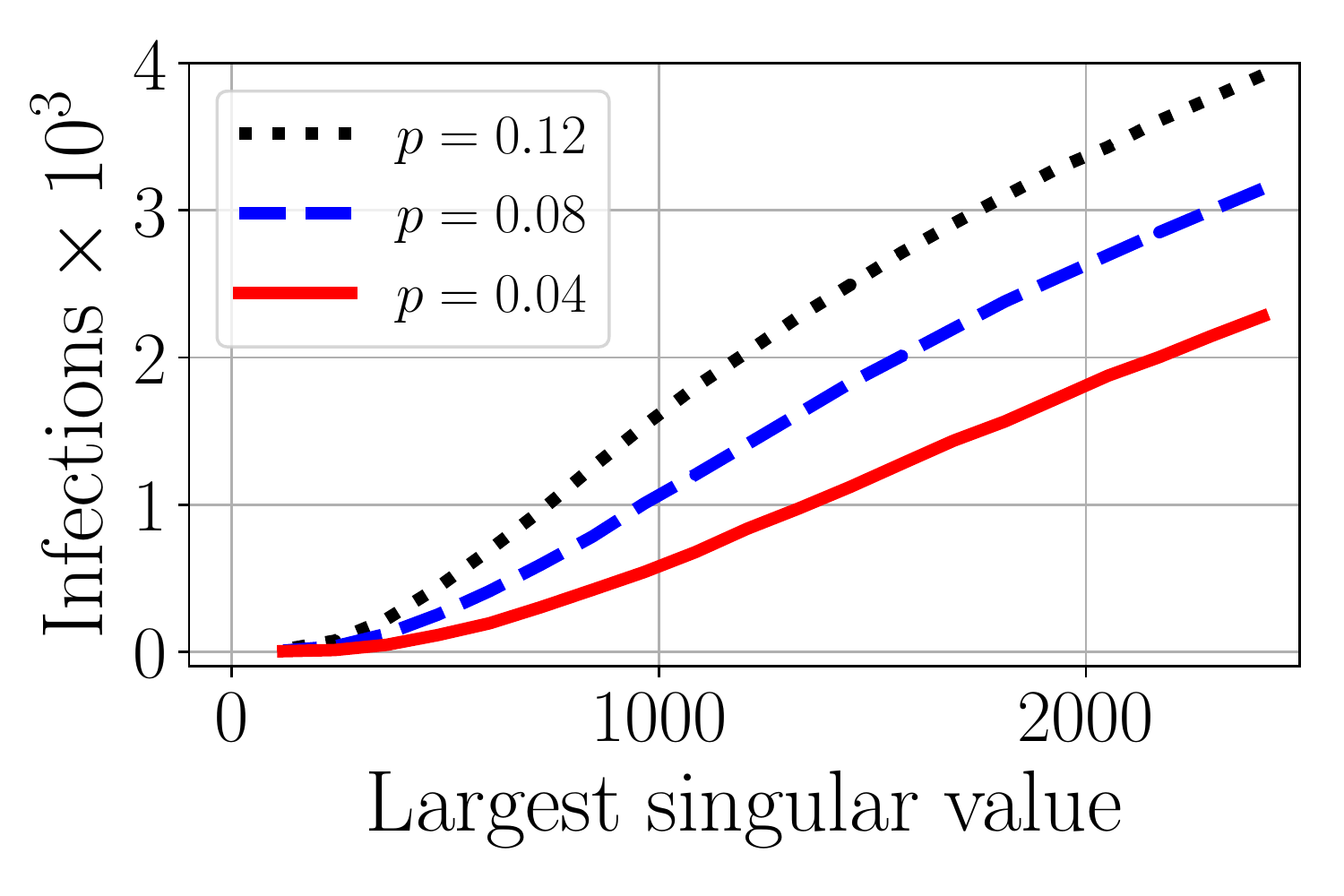}%
  \caption{The number of infections strongly correlates with the largest singular value of the graph: more infections are observed for higher values of $\lambda_1(W)$ (by rescaling $W$). The spreading rate is denoted as $p$.}\label{fig:scale_singular_values}
\end{figure}

\noindent\textbf{Example.}
To give an example of weighted graphs in epidemic spreading, we can consider mobility networks, which describe the movements from groups of individuals to locations.
The graph is weighted by the number of movement records. %
For instance, \citet{chang2021mobility} introduces a mobility-based modeling approach to fit the observed number of infections.
Their approach involves fitting a metapopulation SEIR model with publicly available mobility records.
Recall that an SEIR model uses four compartments to capture a spreading process: Susceptible (S), Exposed (E), Infected (I), and Recovered (R).
In their case, the mobility network is bipartite: one side being points of interest (POIs) and the other being census block groups (CBGs).
One way to convert the weighted bipartite network to our problem setup is by joining the traffic across all POIs for every pair of CBGs via matrix multiplication.

\vspace{0.05in}
\smallskip
\noindent\textbf{Frank-Wolfe algorithm.}
The Frank–Wolfe algorithm is an iterative first-order optimization algorithm for constrained convex optimization (see, e.g., \citet{nocedal2006numerical}).
There are two major steps in the design of this algorithm
First, there is a direction-finding subproblem that computes the descent direction with the smallest correlation with the gradient of the objective.
Second, based on this descent direction, the step size is determined (e.g., by a line search).
Lastly, a gradient descent update is performed using the determined step size and the descent direction.

\section{Spectral Optimization with Frank-Wolfe}\label{sec_alg}

We present a new algorithm to optimize problem \eqref{eq_convex}.
We observe that the gradient of $f(M)$ is equal to the edge centrality scores.
Then, we develop an iterative algorithm with an efficient inner loop that reduces edges with the highest edge centrality.
Lastly, we extend our algorithm to time-varying networks. %

\begin{figure*}[!ht]
	\begin{subfigure}[b]{0.33\textwidth}
		\centering
		\includegraphics[width=0.98\textwidth]{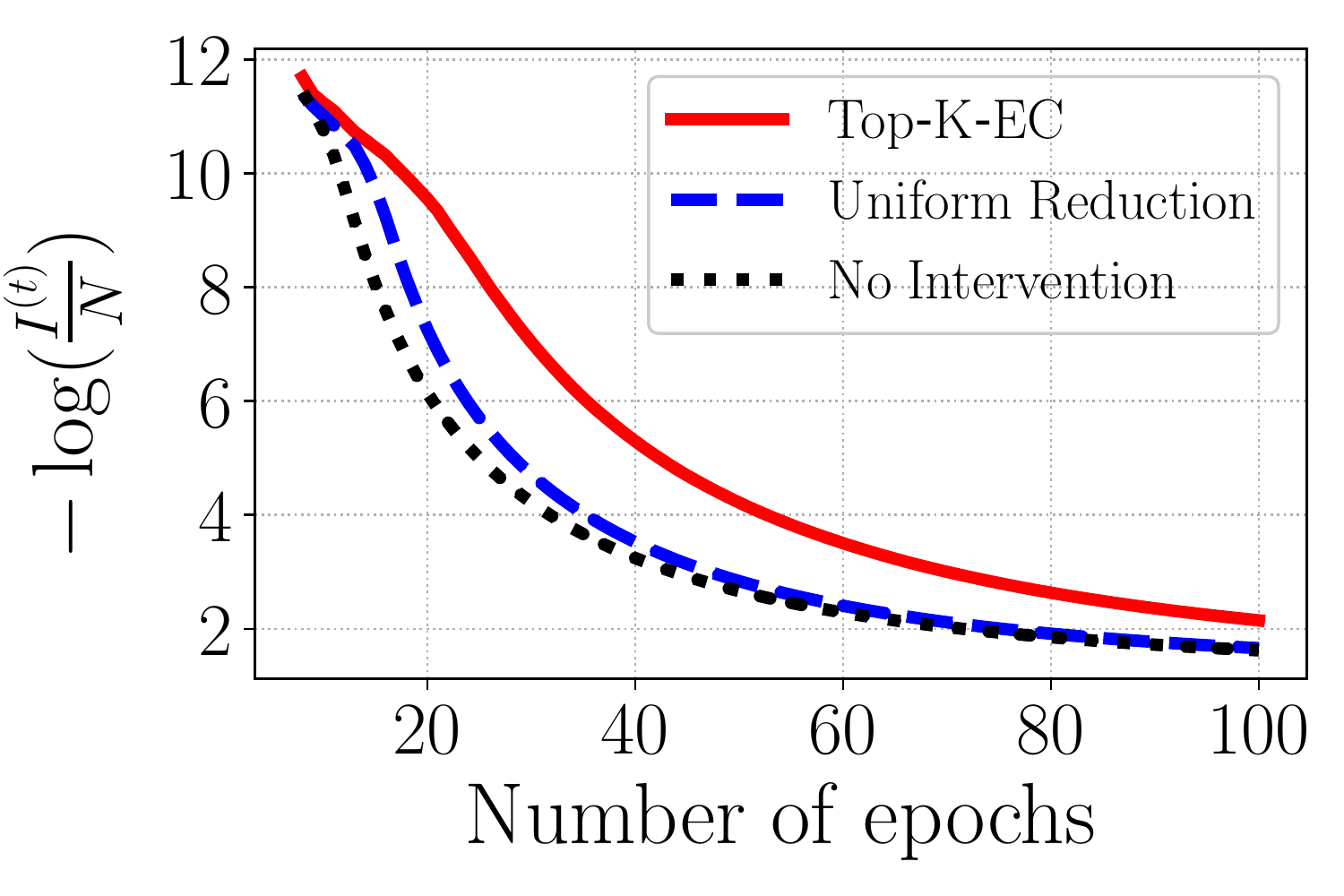}
		\caption{\footnotesize Budget is 1\% of total edge weights.}
		\label{fig_static_comparison1}
	\end{subfigure}\hfill%
	\begin{subfigure}[b]{0.33\textwidth}
		\centering
		\includegraphics[width=0.98\textwidth]{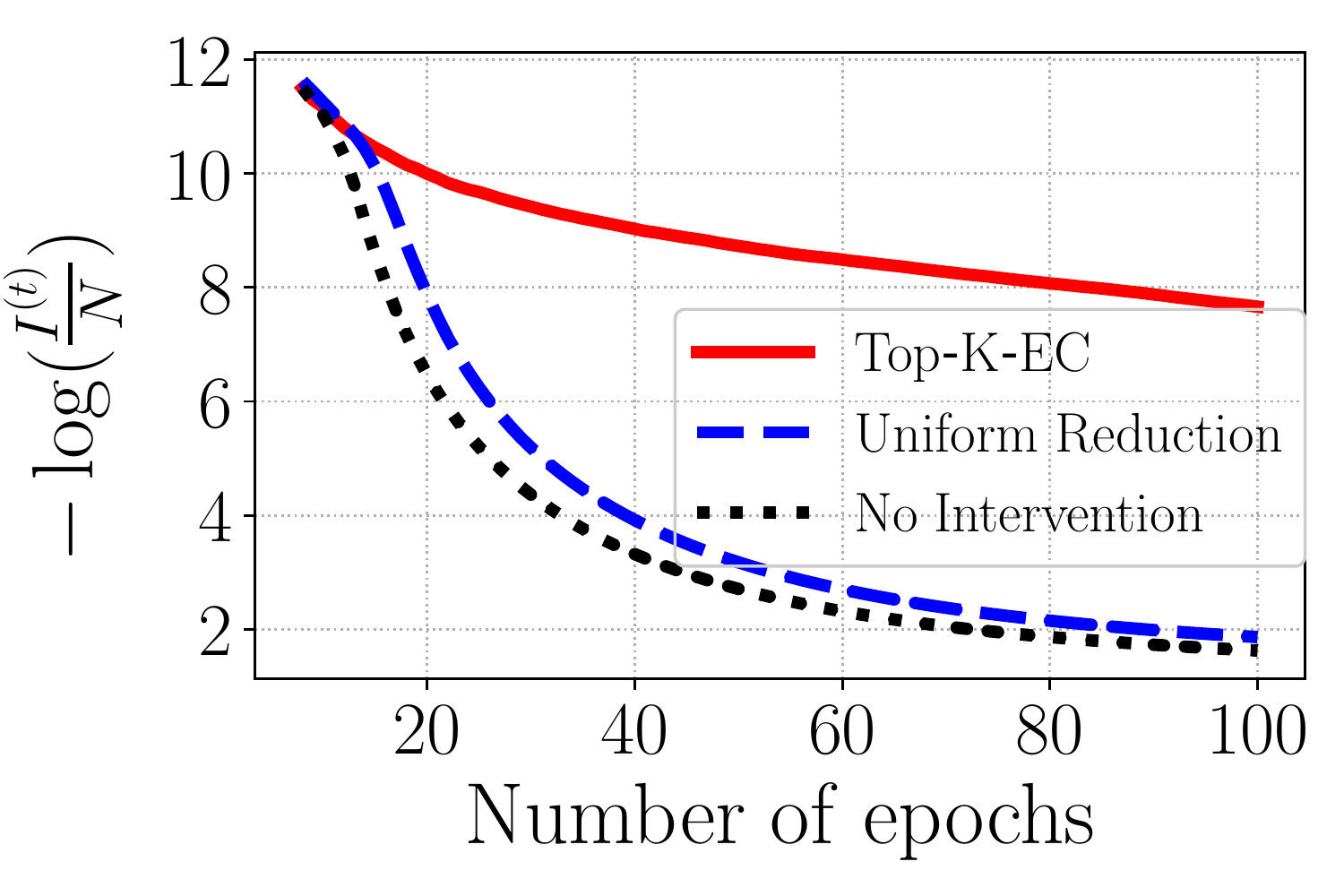}
		\caption{\footnotesize Budget is 20\% of total edge weights.}
		\label{fig_static_comparison20}
	\end{subfigure}\hfill%
	\begin{subfigure}[b]{0.33\textwidth}
		\centering
		\includegraphics[width=0.98\textwidth]{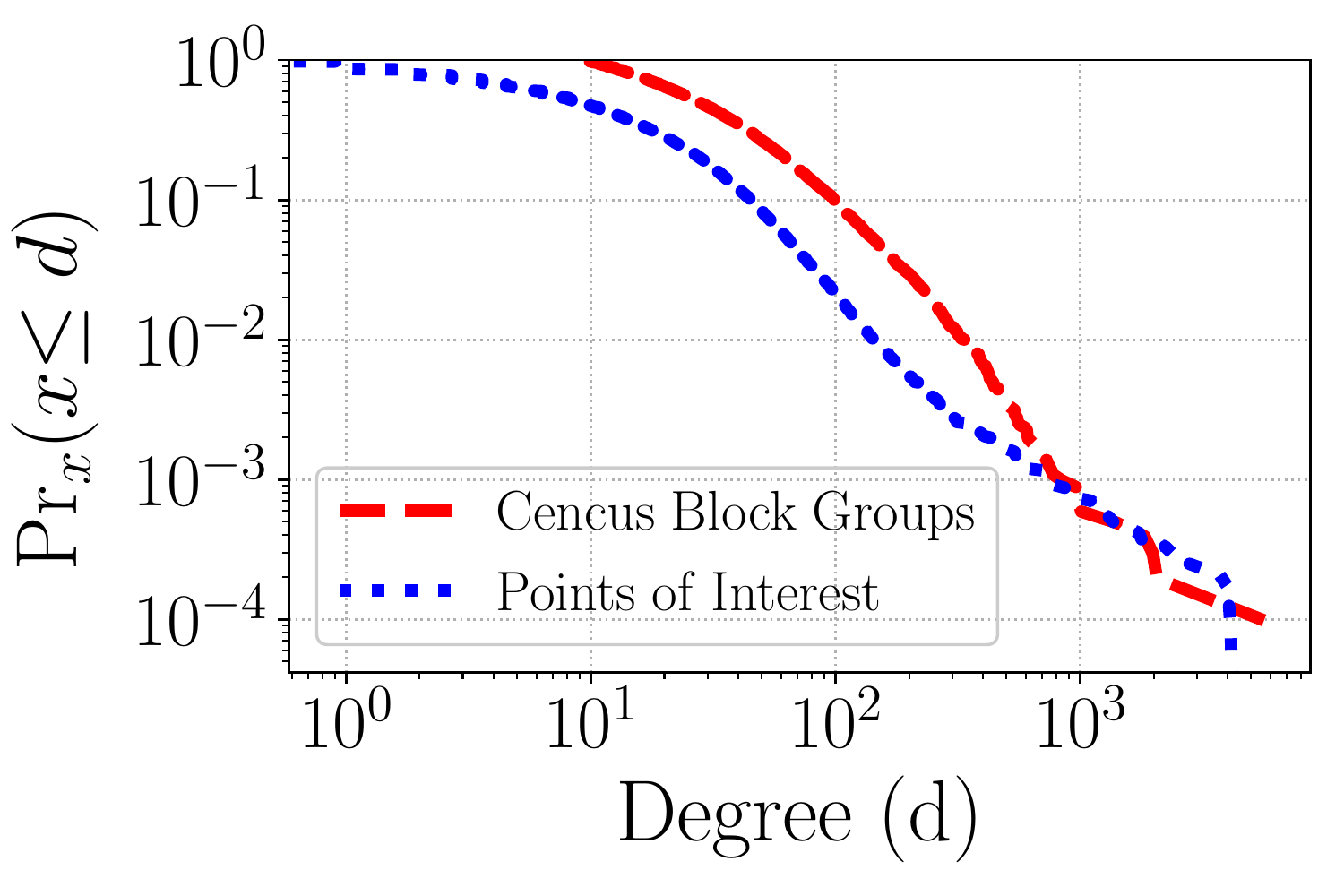}
		\caption{\footnotesize CCDF of nodes.}
		\label{fig_poi_deg}
	\end{subfigure}
	\caption{Comparison of greedy selection and uniform edge-weight reduction on a mobility network.
	Top-K-EC is more effective in reducing the infected proportion throughout the SEIR model simulation.
	Moreover, the groups and the points of interest in the graph follow heavy-tailed degree distributions, supporting our selection using edge centrality scores.}
	\label{fig_sec31}
\end{figure*}

\vspace{-0.04in}
\subsection{Edge centrality as gradient}\label{sec_edge_centrality}

To motivate our approach, we begin by reviewing the approach of \citet{tong2012gelling}, which introduces edge centrality to reduce $f(W)$ for the case of $r = 1$.
The edge centrality score is defined as the product of the eigenvector scores from both ends of an edge.
Let $X$ be any matrix. Let $\vec u_1$ and $\vec v_1$ be the left and right singular vector of $X$, corresponding to $\lambda_1(X)$.
Then, for any edge $(i, j) \in \cE$, its edge centrality score is given by $\vec u_1(i) \cdot \vec v_1(j)$, where $\vec u_1(i)$ denotes the $i$-th coordinate of $\vec u_1$ and $\vec v_1(j)$ denotes the $j$-th coordinate of $\vec v_1$.

The edge-weight reduction can be viewed as a continuous relaxation of edge removal since the weight of an edge can be reduced by a fraction.
Interestingly, we show that the edge centrality scores are equal to the gradient of $\lambda_1(XX^{\top})$ concerning the edge weights up to scaling.
As a result, we generalize edge centrality scores as the gradient of the largest $r$ singular values of $X$.

\begin{lemma}\label{prop_grad}%
    Assume that the singular values of $X$ are all distinct.
    Then, for any $1\le i, j \le n$, the partial derivative of  $(\lambda_1(X))^2$ with respect to $X_{i,j}$ satisfies
    {\small \begin{align}
        \frac{\partial \big((\lambda_1(X))^2\big)}{\partial X_{i, j}} = 2 \lambda_1(X) \cdot \vec u_1(i) \cdot \vec v_1(j). \label{eq_edge_cen_1}
    \end{align}}%
    More generally, for any $r = 1,2,\dots, n$, we have
    {\small
    \begin{align}
        \frac{\partial \big(\sum_{k=1}^r (\lambda_k(X))^2 \big)}{\partial X_{i,j}} = 2 \sum_{k=1}^r \lambda_k(X) \cdot \vec u_k(i) \cdot \vec v_k(j). \label{eq_edge_cen_r}
    \end{align}}
\end{lemma}

Above, $\vec u_k$ and $\vec v_k$ are the left and right singular vectors of $X$ corresponding to $\lambda_k(X)$, and the indices correspond to entries of the vectors.
The proof of Lemma \ref{prop_grad} is presented in Appendix \ref{sec_proof}.
Given a weight matrix $W$ of a network, we compute the edge centrality scores via the best rank-$r$ approximation of $W$ as $\tilde W_r$. %
Let $\big(\tilde W_r\big)_{i, j}$ be the edge centrality score of edge $(i, j) \in \cE$.
We validate that removing edges via top edge centrality scores effectively reduces infections.
Figure \ref{fig_sec31} shows the benefit compared with uniform reduction.

\subsection{Global optimization via iterative greedy}

We now develop the {Frank-Wolfe edge centrality} minimization algorithm, or Frank-Wolfe-EC, specified in Algorithm \ref{alg_edgecen}.
The high-level idea is iteratively applying a greedy selection of edges with the highest generalized edge centrality scores while recomputing the scores:
\begin{itemize}[leftmargin=0.15in]
\setlength\itemsep{0.0em}
\item \textbf{Input:} The primary inputs are graph $\cG$ with weight matrix $W$, an arbitrary budgeted reduction amount $B$, and an arbitrary rank $r \le n$.
\item \textbf{Output:} An $n$ by $n$ weight matrix $M$ with reduced edge weights from $W$.
\end{itemize}

\noindent\textbf{Derivation of the algorithm:} At every iteration $t$ from $1$ to $T$, let $M_t$ be the currently modified weight matrix.
Let $\nabla f(M_t)$ be the gradient of $f(M_t)$.
The Frank-Wolfe algorithm \cite{frank1956algorithm,nocedal2006numerical} computes a descent direction of $M_t$ by minimizing the correlation between the gradient and the iterate subject to the same constraints as problem \eqref{eq_convex}:
{\small
\begin{align}
    G_t^{\star} \leftarrow \arg\min_{X}~~&\quad \inner{X}{\nabla f(M_t)} = 
    \bigtr{\nabla f(M_t)^{\top} X} \label{eq_g_t} \\
    \mbox{s.t.}                 \quad& \sum_{(i,j)\in \cE} \Big(W_{i, j} - X_{i,j}\Big) \le B  \nonumber \\
                                ~~&\quad 0 \le X_{i, j} \le W_{i, j},\,  \forall (i, j)\in \cE, \nonumber \\
                                ~~&\quad X_{i,j}=0,\quad\quad\quad\,\,\, \forall (i, j)\notin \cE. \nonumber
\end{align}}%

The core of our approach is to prove that the optimal descent direction for problem \eqref{eq_g_t} is essentially by removing edges via top edge centrality scores.
Let $X$ be the best rank-$r$ approximation of $M_t$.
Let ${(i_1, j_1), (i_2, j_2), \dots, (i_{m}, j_{m})}$ be the edges in descending order of their generalized edge centrality scores, where $m$ is the number of edges in the graph.
Consider the first $k$ edges whose total weight exceeds the reduction budget $B$.
Then, the weight of the first $k-1$ edges is reduced to zero.
The weight of the last edge decreases with the remaining budget.

Let us call this procedure Top-K-EdgeCentrality (cf. Alg. \ref{alg_edgecen}).
The following result proves that this greedy procedure yields an optimal solution to problem \eqref{eq_g_t}!

\begin{lemma}\label{thm_optimal_descent}
    The optimal solution  $G_t^{\star}$ (cf. \ref{eq_g_t}) is equal to the output of Top-K-EdgeCentrality($W, B; M_t$).
\end{lemma}

\begin{proof}
    By Lemma \ref{prop_grad}, for every edge $(i, j) \in \cE$, the gradient of $f(M_t)$ over this edge is given by the \textit{generalized edge centrality} scores.
    Since $X_{i, j} = 0$ for any $(i, j) \notin \cE$, the optimization objective is:
    {\small \begin{align}
       \inner{X}{\nabla f(M)} = \sum_{(i, j) \in \cE}  2 X_{i, j} \Big( \sum_{k=1}^r  \lambda_k \cdot \vec  u_k(i) \cdot \vec v_k(j) \Big). \label{eq_obj_linear}
    \end{align}}%
    Above, each variable $X_{i,j}$ is multiplied precisely by the generalized edge centrality of the edge $(i, j)$ (cf. line \eqref{alg_ec}).
    Consider minimizing the equivalent objective \eqref{eq_obj_linear} with the constraints of Problem \eqref{eq_g_t}.
    The minimizer, $G_t^{\star}$, is achieved by reducing the weight of the edges with the highest edge centrality to zero until the budget $B$ gets exhausted.
    This is precisely the procedure of Top-K-EC from lines \eqref{alg_ec}-\eqref{alg_ec_red}.
    Thus, we have proved this result.
\end{proof}

After finding the descent direction $G_t^{\star}$, the next step of the Frank-Wolfe algorithm is setting a learning rate $\eta_t$ in a range between $0$ and $1$.
This follows standard procedures from the Frank-Wolfe algorithm \cite{nocedal2006numerical}.
See Algorithm \ref{alg_edgecen} for the complete pseudo-code.

\medskip
\noindent\textbf{Running time analysis:}
Next, we examine the number of iterations needed for Alg. \ref{alg_edgecen} to converge to the global optimum of problem \eqref{eq_convex}.
A well-established result is that the Frank-Wolfe algorithm will converge to the global minimum for convex objectives under mild conditions \cite{nocedal2006numerical}.
Note that objective \eqref{eq_convex} is indeed convex.
Therefore, our algorithm will provably converge to the global minimum of problem \eqref{eq_convex}, denoted as $f^{\textup{OPT}}$.

\begin{theorem}\label{prop_continuous}
    Let $\kappa$ be the minimum of $\lambda_r({M_t}) - \lambda_{r+1}({M_{t}})$ over $t = 0, 1,\dots, T-1$.
    Assume that $\kappa$ is strictly positive.
    Then, the following holds for $M_T$:
    {\small \begin{align}
        f(M_T) - f^{\textup{OPT}} \le \frac{40\Big(\sum_{(i,j)\in\cE} W_{i,j}^2\Big)\alpha_2}{T},\label{eq_converge}
    \end{align}}%
    where $\alpha_2 = \kappa^{-1}{r}^{1/2}\big(\max_{t=1}^T \lambda_1(M_t) \big) + r + C$, for a fixed value $C > 0$.
\end{theorem}

The convergence rate of $O(T^{-1})$ in statement \eqref{eq_converge} is obtained following recent literature (e.g., \cite{jaggi2013revisiting}).
This result guarantees that our algorithm will converge to the global minimum solution under mild conditions.
See Appendix \ref{sec_proof} for the proof.
The constants inherited from the previous guarantee in statement \eqref{eq_converge} can be quite large. However, in our experiments, we observe that less than $30$ iterations are sufficient for the algorithm to converge (at the global optimum).

\medskip
To recap, the running time of our algorithm is $T$ times the running time of each iteration, including:
\begin{itemize}[leftmargin=0.15in]
\setlength\itemsep{0.0em}
    \item Computing a truncated rank-$r$ SVD of a sparse matrix with $m$ nonzeros;
    this requires a time complexity of $O(m r \log (m))$  \cite{musco2015randomized}.
    \item Sorting an array of size $m$; this requires $O(m\log(m))$ time complexity.
\end{itemize}
By comparison, running a linear program solver for problem \eqref{eq_g_t} requires at least $O(m n)$ time complexity \cite{nocedal2006numerical}.
Thus, our approach is most efficient for small $r$.

\subsection{Optimization on time-varying networks}\label{sec_tv}
Our study has focused on mitigating the spread in a static network.
Another consideration is that the network topology evolves over time.
Therefore, an important question is how to tackle such temporal evolution.
Next, we show how to extend our optimization algorithm to time-varying networks.

\medskip
\noindent\textbf{Derivation of the algorithm:} 
Let the weight matrices of a sequence of graphs be denoted as $\cW = \set{W^{(1)}, W^{(2)}, \dots, W^{(s)}}$.
Motivated by the work of \citet{prakash2010virus} which shows the epidemic threshold of time-varying networks, we extend the eigenvalue minimization problem on time-varying networks.
Let $\cM = \set{M^{(1)}, M^{(2)}, \dots, M^{(s)}}$ be a sequence of modified weight matrices.
We aim to find $\cM$ that shrinks the largest eigenvalues of a product matrix:
{\small
\begin{align}
     \min_{\cM}~~&\quad f\big(\cM\big) = \sum_{k=1}^r \Big(\lambda_k\Big(\prod_{t=1}^s M^{(t)}\Big)\Big)^2 \label{eq_convex_temporal} \\
    \mbox{s.t.}              
    ~~&\,\, \sum_{t=1}^s \sum_{(i,j)\in \cE^{(t)}} \big(W_{i, j}^{(t)} - M_{i,j}^{(t)}\big) \le B \nonumber \\
    \quad&\quad 0 \le M_{i, j}^{(t)} \le W_{i, j}^{(t)},\, \forall (i, j)\in \cE^{(t)}, t = 1, \ldots, s,  \nonumber \\
    \quad~~&\quad M_{i,j}^{(t)}=0, \quad\quad\quad\,\,\, \forall (i, j)\notin \cE^{(t)}, t = 1, \ldots, s. \nonumber
\end{align}}%
Above, $\cE^{(t)}$ represents the set of edges in the $t$-th graph of the sequence.
Based on \citet[Theorem 2]{prakash2010virus}, the weight matrix that determines the epidemic threshold process in time-varying networks is the joint product of each weight matrix in the sequence:
$X = \prod_{t=1}^s M^{(t)}$. 
This is why we minimize the largest eigenvalues of the product matrix in $f(\cM)$.

Following Lemma \ref{prop_grad}, we derive the gradient of the largest $r$ eigenvalues of $X^{\top}X$ with respect to $M^{(t)}_{i, j}$, for any $1\le i,j\le n$.
By the chain rule, we have:%
{\small\begin{align}
     \frac{\partial f(\cM)}{\partial M_{i,j}^{(t)}} 
    = \Big\langle{\frac{\partial\Big(\sum_{k=1}^r \big(\lambda_k(X)\big)^2 \Big)}{\partial X}}, {\frac{\partial X}{\partial M_{i,j}^{(t)}}}\Big\rangle. \label{eq_tv}
\end{align}}%
Notice that the first term above on the right is precisely the edge centrality scores we have derived in Lemma \ref{prop_grad}. %
The second term is the product of the rest of the weight matrices in $\cW$ except that $M^{(t)}$ is replaced by an indicator matrix, which is the derivative of $M^{(t)}$ with respect to its $(i, j)$-th entry.
Let $\tilde X_r = U_r D_r V_r^{\top}$ be the rank-$r$ SVD of $X$.
We get (cf. Appendix \ref{sec_proof}):
{\small\begin{align}
    \frac{\partial f(\cM)}{\partial M^{(t)}}
    =2\Big(\prod\nolimits_{k=1}^{t-1} M^{(k)}\Big)^{\top} \tilde{X}_r \Big(\prod\nolimits_{k=t+1}^s M^{(k)}\Big)^{\top}. \label{eq_tv_ec}
\end{align}}%
Matrix \eqref{eq_tv_ec} encodes the edge centrality scores for every edge of $\cE^{(t)}$, at any step $t$.
Thus, we can develop an algorithm for time-varying networks as the static case.
The complete procedure is described in Algorithm \ref{alg_edgecen_temporal}.

\begin{algorithm}[!t]
\caption{Frank-Wolfe for Static Networks}\label{alg_edgecen}
	\begin{footnotesize}
		\begin{algorithmic}[1] %
		    \Input A graph $\cG = (\cV, \cE)$ with weight matrix $W$; Budget $B$.
		    \Param Rank $r$; Iterations $T$; Range of learning rate $H$.
		    \Output A weight matrix $M$ modified from $W$.
			\Procedure{\LP}{$W, B; T, H$}
			    \State Let $M_0 = W$
			    \For {$t = 0, 1, \dots, T-1$}
			        \State $G_t^{\star}$ = \edgecen($W, B; M_t$)
			        \State Set $\eta_t$ by minimizing $f\big((1 - \eta_t) M_t + \eta_t G_t^{\star}\big)$ for $\eta_t \in H$
			        \State $M_{t+1} = (1 - \eta_t) M_t + \eta_t  G_t^{\star}$
			    \EndFor
			    \If {there is unused budget in $M_T$}
			        \State $B' = B - \textup{sum}(W - M_T)$
			        \State $M^{\star}$ = \edgecen($M_T, B'; M_T$)
			    \EndIf
			    \State \Return $M^{\star}$
			\EndProcedure
			\vspace{0.1in}
			\Procedure{\edgecen}{$W, B; M$}
			    \State Let $\tilde M_r$ be the rank-$r$ SVD of $M$ \label{alg_ec}
			    \State Sort the edges in $\cE$ by their edge centrality scores from $\tilde M_r$; let $k$ be the first value such that the total top-$k$ edge weights in $W$ exceed $B$
			    \State Reduce the first $k-1$ edges' weight to zero and the last edge's weight by the remaining budget \label{alg_ec_red}
			    \State \Return the updated $W$
			\EndProcedure	
		\end{algorithmic}
	\end{footnotesize}
\end{algorithm}

\medskip
\noindent\textbf{Running time analysis:}
Similar to Theorem \ref{prop_continuous}, one can then prove that Algorithm \ref{alg_edgecen_temporal} is guaranteed to converge to the optimum solution of problem \eqref{eq_convex_temporal} at the rate of $O(T^{-1})$ after $T$ iterations.
The details of this extension can be found in Appendix \ref{sec_proof}.

\section{Experiments}\label{sec_exp}

We evaluate our proposed approaches on various weighted graphs and mobility networks.
Our experiments seek to address the following questions:
First, does our approach reduce the infections and the largest singular values well compared to methods from prior works? 
Second, what are the effects of each component in our approach, e.g., setting the rank $r$, running iterative greedy selection, and setting the budget?
Third, does our approach run efficiently in practice? 
We present positive results to answer these three questions, validating the practical benefit of our algorithm. {The code repository for reproducing our results can be found online at \url{https://github.com/NEU-StatsML-Research/Designing-Intervention-on-Mobility-Networks}.}

\begin{algorithm}[!t]
	\caption{Frank-Wolfe for Time-Varying Networks}\label{alg_edgecen_temporal}
	\begin{footnotesize}
		\begin{algorithmic}[1]
		    \Input A sequence of graphs with weight matrix $\cW$ in $s$ steps.
            \Param Same as the static case.
		    \Output A sequence of matrices $\cM$ modified from $\cW$.
			\Procedure{\LPTV}{$\cW, B; T, H$}
			    \State Let $\cM_0 = \cW$
			    \For {$t = 0, 1, \dots, T-1$}
			        \State $\mathcal{G}_t = \{G_t^{\star(i)}\}_{i=1}^s$ = \edgecenTV($\cW, B; \cM_t$)
			        \State Set $\eta_t$ by minimizing $f\big((1 - \eta_k)\cM_t + \eta_t \cG_k\big)$ for $\eta_k \in H$
			        \State $\cM_{t+1} = \{M_{t+1}^{(i)} = (1 - \eta_t) M_t^{(i)} + \eta_t  G_t^{\star(i)}: 1\le i\le s\}$
			    \EndFor
			    \If {there is unused budget in $\cM_T$}
			        \State $B' = B - \sum_{i=1}^{s}\textup{sum}(W^{(i)} - M_T^{(i)})$
			        \State $\cM^{\star}$ = \edgecenTV($\cM_T, B'; \cM_T$)
			    \EndIf
			    \State \Return $\cM^{\star}$
			\EndProcedure
			\vspace{0.0325in}
			\Procedure{\edgecenTV}{$\cW, B; \cM$}
			    \State Let $\tilde X_r$ be the rank-$r$ SVD of $X = \prod_{i=1}^{s}M^{(i)}$ \label{alg_ec_temporal}
			    \State Sort the edges in the \emph{union} of $\cE^{(1)}, \cE^{(2)}, \dots, \cE^{(s)}$ by their edge centrality scores (cf. Eq. \ref{eq_tv_ec}); let $k$ be the first value such that the total top-$k$ edge weights from $\cW$ exceed $B$
			    \State Reduce the first $k-1$ edges' weight to zero and the last edge's weight by the remaining budget 
			    \State \Return the updated $\cW$
			\EndProcedure	
	\end{algorithmic}
	\end{footnotesize}
\end{algorithm}

\subsection{Experimental setup}\label{sec_exp_setup}

We use three weighted graphs in our model simulations on static networks: (i) An airport traffic network of flights among commercial airports worldwide.  (ii) A trust network of users on the Advogato platform; (iii) A trust network of users on a Bitcoin platform.
The edge weights in the Airport network denote the number of flight routes between two airports.
Edge weights in the last two networks denote different levels of declared trust among users. The edge weights on the Advogato network are between $0$ and $1$. The edge weights on the Bitcoin network range from $-10$ to $10$. We scale the weights to positive by $\exp(w/5)$. 
The statistics of the networks are listed in Appendix \ref{sec_add_setup}.

\vspace{0.03in}

Besides, we use eight mobility networks constructed with the procedure described in \citet{chang2021mobility}. 
We generate the mobility networks based on the mobility patterns of eight cities. The edge weights denote the population that moves from a group to a location from March 2, 2020, to May 10, 2020. 
The mobility patterns cover 25,341 census block groups with over 65 million people and 147,638 points of interest. We report the statistics of the mobility networks in Table \ref{tab:num_infected}.
We defer a comprehensive discussion of the construction procedure to their paper.

\vspace{0.03in}

We use two sequences of weighted trust networks from Bitcoin-Alpha and Bitcoin-OTC platforms for time-varying networks. Each sequence contains ten trust relationship networks corresponding to five periods. The edge weights are processed in the same way as in the static Bitcoin network.
We also construct time-varying mobility networks corresponding to ten weeks of the same period above for Chicago and Houston.
We describe network data sources in Appendix \ref{sec_add_setup}. 

\begin{table*}[t!]
\centering
\caption{
\textbf{Top:} Dataset statistics for eight mobility networks.
\textbf{Middle:} Comparison of the largest singular value of the edge-weight reduced matrix.
\textbf{Bottom:} Comparison of the total number of infected populations ($\times 10^3$) in SEIR model simulations. 
We report the average number of infections from fifty independent simulations.}\label{tab:num_infected}
\begin{scriptsize}
\begin{tabular}{@{}lcccccccccc@{}}
\toprule
Graphs & AT & CH & DA & HO & MI & NY & PH & DC \\ \midrule
Nodes  & 11,232     & 32,390     & 19,069     & 38,895                  & 17,858        & 34,216   & 18,649 & 10,590             \\
Edges & 154,729   & 439,262   & 283,928   & 671,217                & 276,109      & 463,719  & 260,279                            & 107,733           \\
Avg. edge weight & 5.258 & 4.659 & 4.921 & 4.951 & 4.833 & 4.749 & 4.864 & 4.848 \\
\midrule\midrule
Largest singular value & AT & CH & DA & HO & MI & NY & PH & DC \\ \midrule
No Intervention     &  5526 & 1296 & 2093 & 14677 & 555 & 2413 & 12032 & 1406 \\
Uniform Reduction  &  5250 & 1231 & 1988 & 1394 & 527 & 2292 & 1143 & 1336 \\
Weighted Reduction &  1254 &  302 &  564 &  420 & 213 &  4818 &  374 &  365 \\
Max Capping  &  5250 & 1231 & 1988 & 1394 & 527 & 2292 & 1143 & 1336 \\
POI Category       &  5526 & 1295 & 2073 & 1467 & 555 & 2270 & 1202 & 1375 \\ 
K-EdgeDeletion     &  1565 &  257 &  417 &  447 & 216 &  355 &  282 &  227 \\
\edgecenshort{}    &  1565 &  257 &  417 &  447 & 216 &  355 &  282 &  226 \\
\textbf{Ours (Alg. \ref{alg_edgecen})}         &  \textbf{1191} & \textbf{125} & \textbf{308} & \textbf{235} & \textbf{169} & \textbf{197} & \textbf{190} & \textbf{188} \\\midrule \midrule
Infected populations                & AT & CH & DA & HO & MI & NY & PH & DC \\ \midrule
No Intervention  &  48$\pm$3 & 1858$\pm$46 & 91$\pm$21 & 366$\pm$26 & 752$\pm$26   & 3146$\pm$21 & 492$\pm$20 & 41$\pm$2 \\
Uniform Reduction   &  46$\pm$2 & 1762$\pm$64 & 84$\pm$11 & 312$\pm$26 & 671$\pm$23  & 2996$\pm$40 & 463$\pm$12 & 41$\pm$1 \\
Weighted Reduction    &  43$\pm$2 & 782$\pm$86 & 66$\pm$3 & 194$\pm$18 & 43$\pm$12  & 1336$\pm$60 & 342$\pm$10 & 40$\pm$1 \\
Max Capping    &  44$\pm$2 & 1741$\pm$65 & 82$\pm$8  & 315$\pm$33 & 675$\pm$26   & 2990$\pm$45 & 455$\pm$15 & 41$\pm$1 \\
POI Category    &  46$\pm$3 & 1728$\pm$62 & 77$\pm$8  & 283$\pm$31 & 687$\pm$25   & 2950$\pm$38 & 458$\pm$17 & 41$\pm$1 \\
K-EdgeDeletion & 44$\pm$2 & 346$\pm$40 & 64$\pm$2 & 186$\pm$18 & 78$\pm$8 & 352$\pm$27 & 185$\pm$10 & 39$\pm$1 \\
\edgecenshort{} &  45$\pm$3 & 355$\pm$46 & 64$\pm$2 & 187$\pm$21 & 78$\pm$7 & 362$\pm$36 & 178$\pm$11 & 39$\pm$1 \\
\textbf{Ours (Alg. \ref{alg_edgecen})}      &  \textbf{40$\pm$1} & \textbf{166$\pm$16} & \textbf{62$\pm$2} & \textbf{86$\pm$10} & \textbf{8$\pm$2} & \textbf{301$\pm$88} & \textbf{129$\pm$13}	& \textbf{39$\pm$1} 
\\\bottomrule
\end{tabular}
\end{scriptsize}
\end{table*}

\smallskip
\noindent\textbf{Baseline methods.}
The experiments of spreading on static networks involve the following baseline methods: 
(1) K-EdgeDeletion: Delete a set of edges with the highest edge centrality scores according to the best rank-1 approximation of $W$ \cite{tong2012gelling}. 
(2) Weighted reduction: Reduce the weight of every edge by a ratio that is proportional to its weight. (3) Uniform reduction: Uniformly reduce the weight of every edge by the same fraction.
(4) Max occupancy capping: Reduce the cumulative weights at each POI proportional to its max occupancy.
(5) Capping by POI category: Cap the maximum occupancy of a particular category of POIs.
The last three baselines are adapted from \citet{chang2021mobility}.

We consider a similar set of baseline methods for time-varying networks, including uniform reduction,  weighted reduction, and the K-EdgeDeletion method \cite{tong2012gelling}. 
The difference from methods on static networks is that edge weight reduction strategies are applied to all edges in the sequence of networks. 


\medskip
\noindent\textbf{Implementation.}
We simulate an SEIR model on each weighted network.
On weighted graphs, a node can get infected by its infectious neighbors with a probability equal to the edge weight times the transmission rate. 
We use a transmission rate of 0.05 and an initially exposed ratio of 0.01.
We follow the procedure of \citet{chang2021mobility} on mobility networks to simulate a metapopulation SEIR model in each network where one SEIR model is instantiated for each CBG. 
We calibrate the parameters of SEIR models so that the simulated cases approximate the reported cases from New York Times COVID-19 data. 
Besides, we also evaluate our algorithm on other variants of epidemic models, including SIR and SIS with the same parameters. 
We describe the simulation setup details in Appendix \ref{sec_add_setup}.
For completeness, a brief description of the epidemic models is provided in Appendix \ref{sec_epi_models}. 

In Algorithm \ref{alg_edgecen} and \ref{alg_edgecen_temporal}, we search the rank parameter $r$ in $[1, 50]$ and the number of iterations in $[5, 30]$. For each result reported in Section \ref{sec_exp}, we search the two hyper-parameters 50 times. 
We use an edge-weight reduction budget of 5\% of the total edge weights. Results of using other budget amounts are consistent and are discussed in Section \ref{sec_ablation}.
We use 30 values from the range of $[10^{-3}, 10^{-1}]$ as the range of learning rate $H$. 
For weighted graphs, we directly use the weight matrix as $W$. We compose the weight matrix $W$ for mobility networks by multiplying the bipartite network matrix and its transpose.
All the experiments are conducted on an AMD 24-Core CPU machine.

\subsection{Experimental results}\label{sec:main_results}
Our algorithms effectively control infections by reducing the largest singular value on a range of static and time-varying networks. We observe consistent results across various epidemic models, including SEIR, SIR, and SIS.
\vspace{-0.025in}
\begin{itemize}[leftmargin=0.15in]
\setlength\itemsep{0.0em}
\item {\bf Drop in the largest singular value:}
Figure \ref{fig_intro} illustrates the largest singular value of the modified weight matrix of the three weighted graphs.
\LPshort{} reduces the largest singular value more than baselines by \textbf{11.4}\% on average.
Additionally, Table \ref{tab:num_infected} reports the largest singular value of modified mobility networks.
\LPshort{} is $\mathbf{30.7\%}$ more effective than the best baseline  on average.

\item {\bf Reduced number of infections:}
Figure \ref{fig_intro} compares our algorithm to baseline intervention strategies on three weighted graphs.
Overall, our algorithm reduces the number of infected nodes by $\mathbf{10.4\%}$ more than baselines on average.
Table \ref{tab:num_infected} compares the total number of infected populations on eight mobility networks.
Note that ours outperform other baselines by $\mathbf{30.1\%}$ on average and up to $\mathbf{80.3\%}$.

\item {\bf Results for time-varying networks:}
On time-varying networks, \LPTVshort{} also outperforms other baselines.
The number of infections is smaller by $\mathbf{6.9\%}$ averaged over both  time-varying weighted graphs and mobility networks.

\item {\bf Simulation using SIS and SIR:} 
Our approach also helps reduce infections in SIR and SIS epidemic models. We observe that \LPshort{} reduces the number of infections by $\mathbf{14.7\%}$ and $\mathbf{10.8\%}$ more on average over the eight static mobility networks.
\end{itemize}

\subsection{Ablation studies}\label{sec_ablation}
%
We ablate the parameters in our approach and provide further insights into the properties of our algorithm. 
\begin{itemize}[leftmargin=0.15in]
\setlength\itemsep{0.0em}

\item {\itshape Benefit of choosing ranks:}
Recall that our algorithm requires specifying the rank $r$--the number of top singular values--in Equation \ref{eq_convex}. 
We hypothesize that varying the rank $r$ would lead to different intervention results. 
We ablate the performance of our algorithm by using different $r$ in a range of $[1, 50]$. The results show that the performance of the best choice $r$ outperforms using $r=1$ by \textbf{40.2\%} averaged over all networks.
This result justifies our formulation of the network intervention problem as an optimization for the sum of largest-$r$ singular values instead of only the largest single value.

\item {\itshape Benefit of being iterative:} The greedy selection algorithm \edgecenshort{} can be viewed as a special case of  \LPshort{} with $T = 1$.
Notice that our iterative approach is necessary to achieve the observed performance.
In Table \ref{tab:num_infected},  \LPshort{} outperforms \edgecenshort{} by \textbf{31.4}\% on average, and the largest singular value is reduced by \textbf{33.1}\% more.

\item {\itshape Varying budget $B$:}
We have also observed similar results by varying the budget for mobility reduction. 
We vary the budget from 1\% to 20\% using the New York mobility network.
Our algorithm outperforms the baselines consistently using different budgets, similarly for the largest singular value.
Interestingly, when the budget level is small (e.g., 1\%), \LPshort{} reduces the largest singular value more significantly than baseline methods.
\end{itemize}

\subsection{Runtime report}\label{sec:scalability}
%
Across all eleven graphs, our approach converges within 30 iterations (or 17 on average).
Each iteration requires an SVD step that takes less than 3 seconds.
The other steps in each iteration require less than 2.7 seconds.
For larger graph instances, we run our method on seven graphs with the number of edges included: com-Orkut (117M), com-LiveJournal (34M), wiki-topcats (28M), web-BerkStan (7.6M), web-Google (5.1M), web-Stanford (2.3M), and web-NotreDame (1.4M) from the SNAP datasets.
Figure \ref{fig:runtime_per_iteration} reports the runtime for one iteration of our algorithm.
Notice that the runtime scales almost linearly with the number of edges. Our algorithm takes 4943 seconds on the largest graph with 117M edges and 3M nodes.
These results show that our algorithm runs efficiently on large-scale graphs.

\begin{figure}[!h]
	\centering
	\includegraphics[width=0.3234\textwidth]{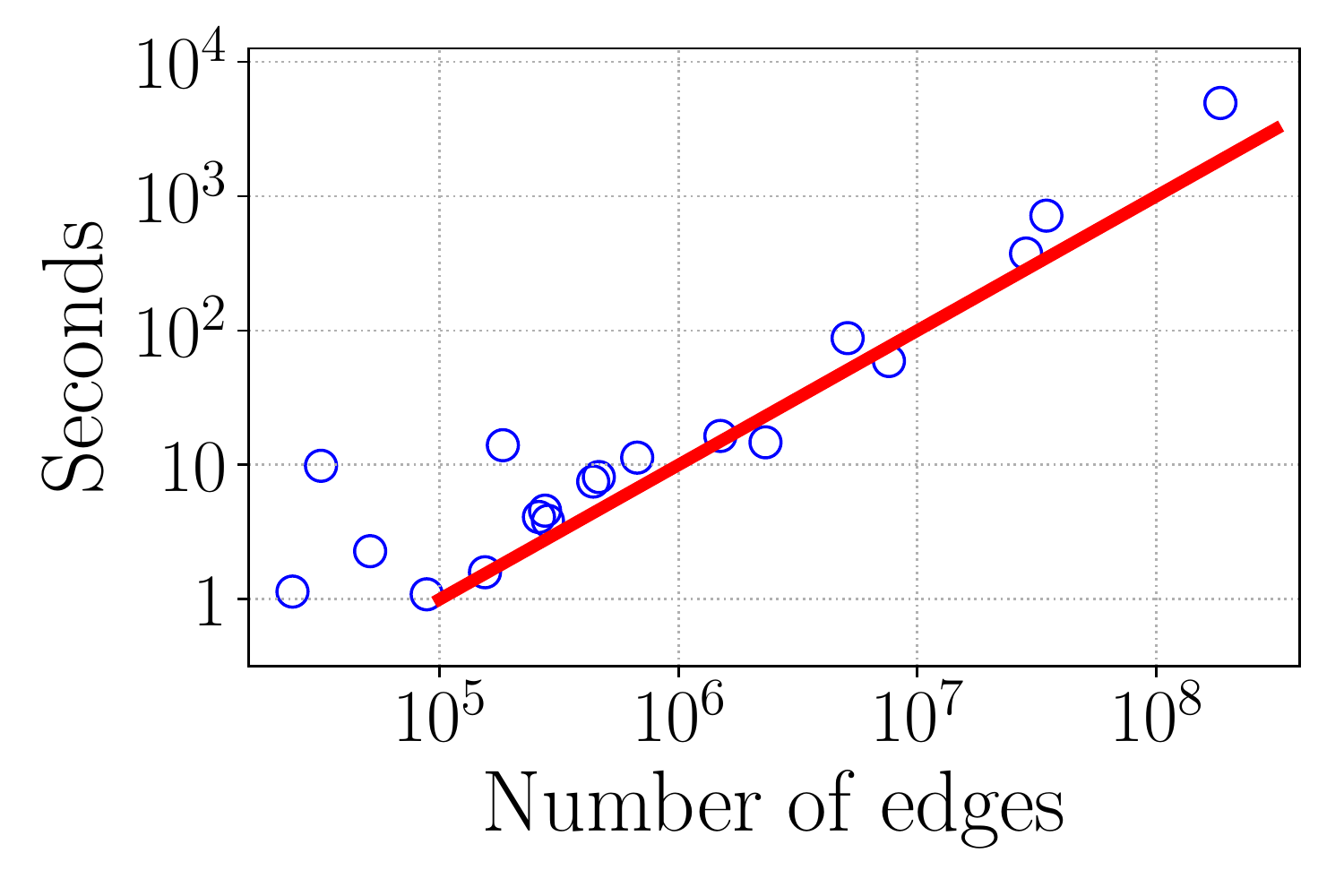}
	\caption{Runtime of Frank-Wolfe-EC in log-log scale for one iteration. The number of  edges ranges from $10^4$ to $10^8$, and the number of  nodes ranges from $10^3$ to $10^6$.}
	\label{fig:runtime_per_iteration}
\end{figure}
\section{Related Work}\label{sec_related}

There is a significant amount of work on diffusion processes on networks.
A detailed survey from an epidemic perspective can be found in  \citet{pastor2015review}.
A key result in the literature is that the largest eigenvalue of the adjacency matrix (a.k.a.~the spectral radius) characterizes the epidemic threshold for many propagation models \cite{wang2003epidemic,ganesh2005effect,prakash2012threshold}.
An important implication of this result is that the epidemic dies out if the spectral radius decreases, and this has motivated many works on epidemic control \cite{van2011decreasing,le2015met,chen2016eigen}.
Because eigen-optimization problems via edge additions or deletions are NP-hard \cite{khalil2014scalable}, both heuristic solutions and principled approximation algorithms have been investigated.
A practical approach in the literature is following the greedy algorithm with a node centrality \cite{chen2015node} or edge centrality notion \cite{tong2012gelling} (see also several alternative edge centrality notions in link recommendation \cite{parotsidis2016centrality} and distance sketching \cite{zhang2019pruning}).
Our notion of edge centrality follows the edge centrality notion studied in \citet{tong2012gelling}.
A related literature studies diffusion control in the Firefighter problem \cite{anshelevich2009approximation}.
Besides epidemic spreading, diffusion processes are also studied in social networks (e.g., \cite{matsubara2012rise,goel2015note,haghtalab2017monitoring}), and financial transaction networks \cite{goel2014connectivity}.

Our work applies the Frank-Wolfe algorithm, a classic algorithm for constrained optimization  \cite{frank1956algorithm,nocedal2006numerical} to study graph spectral optimization.
The Frank-Wolfe algorithm and its theoretical property are well-studied in the machine learning and optimization literature (see, e.g., \citet{jaggi2013revisiting}, \citet{tajima2021frank}, and the references therein). %
We observe a connection between edge centrality and gradients which significantly speeds up the Frank-Wolfe algorithm compared with a naive implementation using a linear program solver.
One relevant application for our approach is to consider node-level intervention measures.
For mobility networks, reducing the weight of a node means restricting a particular group or location's mobility.
Our approach can naturally extend to node-level reduction by similarly deriving node centrality scores as gradients.
Besides, there are also methods for speeding up eigenscore computation on dynamic graphs \cite{chen2017eigen,zhang2016approximate}. It is conceivable that one could combine this method with our approach to achieve the best of both worlds.
Finally, there are studies on the design of vaccine distribution for pandemic control \cite{zhang2014scalable,sambaturu2020designing} and optimization for network robustness \cite{chan2016optimizing}.
It would be interesting to use the new tools developed in this paper to study these related problems.

\section{Conclusion}\label{sec_discuss}

This work considered controlling diffusion processes on weighted graphs.
We study minimizing the largest eigenvalues of the graph and design an efficient algorithm that is guaranteed to converge to the global minimum.
We observe a connection between edge centrality scores and gradients, which provides a new way to derive graph spectral optimization algorithms.
We show how to derive them for static and time-varying networks.
Experiments show that our algorithms are effective on various epidemic models and weighted graphs.

We mention two open questions for future work.
First, it would be interesting to understand better the metapopulation SEIR model of \citet{chang2021mobility} such as its epidemic threshold.
Second, it would be interesting to understand better how eigenvalues affect the diffusion process besides $\lambda_1$. 
We hope our work inspires further algorithmic and theoretical studies about epidemics.

\section*{Acknowledgement}

Thanks to Bijaya Adhikari for bringing reference \cite{prakash2010virus} to the authors' attention.
Thanks to Aditya Prakash for the helpful discussions.
We thank the anonymous referees for their constructive feedback.
DL acknowledges the financial support from the startup fund and a seed/proof-of-concept grant from the Khoury College of Computer Sciences, Northeastern University.

\begin{refcontext}[sorting=nyt]
\printbibliography
\end{refcontext}

\appendix
\clearpage
\onecolumn
\section{Complete Proofs}\label{sec_proof}

This section lays out the proofs for our statements in Section \ref{sec_alg}.
For ease of reading, we include a list of notations needed in the proofs below.

\begin{table}[h]
\centering
\caption{A table of notations used in our paper for reference.}\label{table_notations}
{\footnotesize \begin{tabular}{@{}ll@{}}
        \toprule
        Symbol & Definition \\ \midrule
$\cG = (\cV, \cE)$     &  Weighted and possibly directed graph      \\ 
$W$             &    A nonnegative weight matrix of $\cG$    \\
$W_{i,j}$       & The $(i, j)$-th entry of $W$ \\
$\lambda_k(W)$  &  The $k$-th largest singular value of a matrix $W$ \\
$\vec u_k$ & The left singular vector of a weight matrix $W$ corresponding to $\lambda_k(W)$ \\
$\vec v_k$ & The right singular vector of a weight matrix $W$ corresponding to $\lambda_k(W)$ \\
$\vec v(i)$ & The $i$-th coordinate of the vector $\vec v$ \\
$\tilde X_r$    & The best rank-$r$ approximation of $X$ \\
$\cW$          &  A sequence of weight matrices from timestamp $1$ to $s$ \\
$\cE^{(t)}$     &   The set of edges in the $t$-th graph of the sequence \\
$W^{(t)}$       &   A nonnegative square weight matrix for the graph at timestamp $t$ \\
$\norm{\cdot}$    & The $\ell_2$ norm of a vector or the spectral norm of a matrix \\
$\bignormFro{\cdot \ }$   & The Frobenius norm of a matrix  \\
$\langle \cdot , \cdot \rangle$  & The matrix inner product between two matrices 
            \\\bottomrule
\end{tabular}}
\end{table}

\subsection{Proof for the iterative greedy algorithm}
First, we prove the connection between generalized edge centrality and gradients.

\medskip
\noindent\textbf{Proof of Lemma \ref{prop_grad}.}
    Consider a singular value $\lambda_k$ of $X$, for any $k$.
    Let $\vec u_k$ and $\vec v_k$ be the left and right singular vectors of $X$ corresponding to $\lambda_k$, respectively.
    By the chain rule, it suffices to show that $\frac{\partial \lambda_k(X)}{\partial X_{i,j}} = \vec u_k(i) \cdot \vec v_k(j)$.
    First, we have
    $\vec u_k^{\top} X = \lambda_k  \vec v_k^{\top}.$
    We differentiate over $X$ on both sides of the above equation:
    \begin{align}
        \der(\vec u_k^{\top}) X + \vec u_k^{\top} \der(X) = \der(\lambda_k) \vec v_k^{\top} + \lambda_k \der(\vec v_k^{\top}). \label{eq_partial}
    \end{align}
    Since $\vec v_k$ is a unit length vector,
    \begin{align}
        \der(\norm{\vec v_k}^2) = 2\inner{\vec v_k}{\der(\vec v_k)} = 2\der(\vec v_k^{\top}) \vec v_k = 0. \label{eq_orth}
    \end{align}
    Thus, by multiplying both sides of equation \eqref{eq_partial} with $\vec v_k$, we  get
    \begin{align}
        \der(\vec u_k^{\top}) X \vec v_k + \vec u_k^{\top} \der(X) \vec v_k
        = \der(\lambda_k) \vec v_k^{\top} \vec v_k + \lambda_k \der(\vec v_k^{\top}) \vec v_k, \label{eq_multiply}
    \end{align}
    which is equal to $\der(\lambda_k)$ since equation \eqref{eq_orth} holds and $v_k$ is a unit length vector.
    Looking at equation \eqref{eq_multiply}, we observe
    \begin{align}
        \der(\vec u_k^{\top}) X \vec v_k = \der(\vec u_k^{\top}) \lambda_k \vec u_k = \lambda_k \der(\vec u_k^{\top}) \vec u_k = 0, \label{eq_i_j}
    \end{align}
    where the last step follows similarly to equation \eqref{eq_orth}, since $\vec u_k$ is also a unit length vector.
    In summary, we have shown
    $\vec u_k^{\top} \der(X) \vec v_k = \der(\lambda_k)$.
    This implies that the derivative of $\lambda_k$ over $X_{i,j}$ is equal to $\vec u_k(i) \cdot \vec v_k(j)$.
    Since this holds for any $k$, we thus conclude that equations \eqref{eq_edge_cen_1} and \eqref{eq_edge_cen_r} are both true.\hfill$\square$

\subsection{Proof for the running guarantee}

Next, we derive the convergence guarantee of Algorithm \ref{alg_edgecen}.

\medskip
\noindent\textbf{Proof of Theorem \ref{prop_continuous}.}
    We complete the convergence analysis of our algorithm.
    First, we show that the objective function $f(M)$ is convex in $M$.
    Second, we invoke the result of \citet{jaggi2013revisiting}, specifically Lemma 7 and Theorem 1, which show that as long as the gradient $\nabla f(M)$ is Lipschitz-continuous and the constraint set has bounded diameter, the Frank-Wolfe algorithm will converge to the optimum at a rate of $O(\frac 1 t)$ after $t$ iterations.

    We first show that the sum of top singular values $g(M) = \sum_{k=1}^r \lambda_k(M)$ is convex.
    With the variational characterization of singular values, $g(M)$ is equal to
    \begin{align}\label{eq_char}
        g(M) = \max_{U^{\top}U = V^{\top} V = \id_r:~U\in\real^{n\times r}, V\in\real^{m\times r}} \inner{UV^{\top}}{M}.
    \end{align}
    Thus, for any $n$ by $m$ matrix $M_1,M_2$, and any $\alpha \in [0, 1]$, let $\tilde U$ and $\tilde{V}$ be the maximizer of the above for $f\big(\alpha M_1 + (1 - \alpha) M_2\big)$.
    Therefore,
    \begin{align*}
        g\big(\alpha M_1 + (1 - \alpha) M_2\big)
        &= \inner{\tilde U \tilde V^{\top}}{\alpha M_1 + (1- \alpha) M_2} \\
        &\le \alpha \inner{\tilde U \tilde V^{\top}}{M_1} + (1-\alpha) \inner{\tilde U \tilde V^{\top}}{M_2} \\
        &\le \alpha g(M_1) + (1 - \alpha) g(M_2),
    \end{align*}
    which implies that $g(M)$ is convex. Next, we show that $f(M)$ is convex.
    For any $\alpha \in [0, 1]$,
    \begin{align*}
        &f(\alpha M_1 + (1- \alpha) M_2)
        = g\Big( (\alpha M_1 + (1- \alpha) M_2)^T(\alpha M_1 + (1- \alpha) M_2) \Big)\\
        &\le \alpha^2 g(M_1^{\top}M_1) + (1-\alpha)^2 g(M_2^{\top}M_2)
        + 2\alpha(1-\alpha) g(M_1^{\top}M_2).
    \end{align*}
    Let $\tilde U$ and $\tilde V$ be the maximizer of \eqref{eq_char} for $M_1^{\top} M_2$. We have
    \begin{align*}
        &2g(M_1^{\top} M_2)
        = 2\inner{\tilde U \tilde V^{\top}}{M_1^{\top} M_2}
        = 2\inner{M_1 \tilde U}{M_2 \tilde V} \\
        \le& \bignormFro{M_1 \tilde U}^2 + \bignormFro{M_2 \tilde V}^2 = \inner{M_1^{\top} M_1}{\tilde U \tilde U^{\top}} + \inner{M_2^{\top} M_2}{\tilde V \tilde V^{\top}} \\
        \le& g(M_1^{\top} M_1) + g(M_2^{\top} M_2).
    \end{align*}
    Therefore, $f(\alpha M_1 + (1-\alpha)M_2)$ is less than $\alpha \cdot g(M_1^{\top} M_1) = \alpha \cdot f(M_1)$ plus $(1-\alpha)\cdot g(M_2^{\top} M_2) = (1-\alpha) \cdot f(M_2)$.

    Second, we verify that $\nabla f(M)$ is $\alpha_2$ Lipschitz continuous in the Frobenius norm.
    The proof is based on matrix perturbation bounds.
    Let $\tilde M = M + E$ be a perturbation of $M$.
    Let $M_r = U_r D_r V_r^{\top}$ be the top-$r$ SVD of $M$.
    Let $\mu_1$ be the largest singular value of $M$.
    Let $\tilde M_r = \tilde U_r \tilde D_r \tilde V_r^{\top}$ be the top-$r$ SVD of $\tilde M$.
    First, consider $\norm{E}_2 \le \kappa / 2$.
    By matrix perturbation bounds on the truncated SVD of a matrix (e.g., Theorem 1 of \citet{vu2021perturbation}; the condition is satisfied since $\kappa$ is the spectral gap between the $r$-th and $(r+1)$-th largest singular values), we have
    \begin{align*}
        \normFro{M_r - \tilde M_r}^2
        \le 2\normFro{E}^2 + \frac{4\lambda_1^2}{\kappa^2}\normFro{E}^2 + C\normFro{E}^2.
    \end{align*}

     When $\norm{E}_2 \ge \kappa/2$, notice that
     \begin{align*}
         \normFro{M_r - \tilde M_r}^2 &= \normFro{U_r D_r V_r^{\top} - \tilde U_r \tilde D_r \tilde V_r^{\top}}^2 \\
         &\le 2\normFro{D_r}^2 + 2\normFro{\tilde D_r}^2 \\
         &\le 2r \lambda_1^2 + 2r(\lambda_1 + \norm{E}_2)^2,
     \end{align*}
     which is at most $2r(3\lambda_1^2 +  2\norm{E}_2^2)$.
     The step above uses the Weyl's Theorem that $\norm{D_r - \tilde D_r}_2 \le \norm{E}_2$.
     Taken together, we conclude that $\nabla f(M)$ must be
     \[ \sqrt{\max\Big(2 + \frac{4 \lambda_1^2}{\kappa^2} + C, \frac{24r \cdot \lambda_1^2}{\kappa^2} + 4r \Big)} \]
     Lipschitz-continuous.
     Lastly, the diameter of the constraint set is at most $\sqrt{\sum_{(i,j)\in\cE} W_{i,j}^2}$, since for every $(i,j)\in\cE$, the search space is bounded between $0$ and $W_{i,j}$.
     Taken together, we have proved that:
     $f(M)$ is convex, $\nabla f(M)$ is $\alpha_2$ Lipschitz continuous, and the diameter of the constrained space of problem \eqref{eq_convex} is $\sqrt{\alpha_1/8}$.
     Using Lemma 7 and Theorem 1 of \citet{jaggi2013revisiting}, the proof is complete.\hfill$\square$

\medskip
\noindent\textbf{Extension to time-varying networks.} Notice that the time-varying extension is a special case of the above result.
Therefore, the same convergence rate of $O(T^{-1})$ holds for Algorithm \ref{alg_edgecen_temporal} towards the global optimum of problem \eqref{eq_convex_temporal}.

\medskip
Lastly, we derive the gradient of the largest $r$ eigenvalues of $X^{\top}X$ where $X$ is the product of the weight matrices in the sequence of time-varying networks (cf. Section \ref{sec_tv}).

\medskip
\subsection{Derivation of the time-varying case: Equation \ref{eq_tv_ec}.} Let $\set{M^{(1)}, M^{(2)}, \dots, M^{(s)}}$ be a sequence of modified weight matrices and $X = \prod_{t=1}^{s}M^{(t)}$. Following Lemma \ref{prop_grad}, we derive the gradient of the largest $r$ eigenvalues of $X^{\top}X$ with respect to $M^{(t)}_{i, j}$, for any $1\le i,j\le n$.
By the chain rule, we have:
{\small\begin{align}
     \frac{\partial f(\cM)}{\partial M_{i,j}^{(t)}} 
    = \Big\langle{\frac{\partial\Big(\sum_{k=1}^r \big(\lambda_k(X)\big)^2 \Big)}{\partial X}}, {\frac{\partial X}{\partial M_{i,j}^{(t)}}}\Big\rangle. \label{eq_tv}
\end{align}}

\noindent Notice that the first term above on the right is precisely the edge centrality scores we have derived in Lemma \ref{prop_grad}.
The second term is essentially the product of the rest of the weight matrices in $\cW$ except that $M^{(t)}$ is replaced by an indicator matrix, which is the derivative of $M^{(t)}$ with respect to its $(i, j)$-th entry.

Let $\tilde X_r = U_r D_r V_r^{\top}$ be the rank-$r$ SVD of $X$. Let the product of weight matrices from $1$ to $t-1$ as $A = \prod_{k=1}^{t-1} M^{(k)}$ and the product of weight matrices from $t+1$ to $s$ as $B = \prod_{k=t+1}^{s} M^{(k)}$. $A$ is equal to identity matrix when $t = 1$, and $B$ is equal to identity matrix when $t = s$. Let $J^{i, j}$ as a single-entry indicator matrix where its $(i, j)$-th entry is $1$, and the rest of the entries are equal to 0. Then, we can rewrite the gradient as follows:
{\small\begin{align}
     \frac{\partial f(\cM)}{\partial M_{i,j}^{(t)}} 
    = 2 \Big\langle{\tilde X_r}, {AJ^{i,j}B}\Big\rangle 
    = 2 \sum_{1 \leq p,q \leq n} {\big( \tilde X_r \big )}_{p,q} {\big( AJ^{i,j}B \big)}_{p,q}
    = 2 \sum_{1 \leq p,q \leq n} {\big( \tilde X_r \big )}_{p,q} A_{p, i} B_{j, q}  
    = 2 \Big( A^{\top} {\tilde X_r } B^{\top} \Big)_{i, j}
\end{align}}

\noindent Thus, we get the gradient of $f(\cM)$ with respect to the weight matrix $M^{(t)}$ as follows:
{\small\begin{align}
    \frac{\partial f(\cM)}{\partial M^{(t)}}
    =2 A^{\top} {\tilde X_r } B^{\top}
    =2\Big(\prod\nolimits_{k=1}^{t-1} M^{(k)}\Big)^{\top} \tilde{X}_r \Big(\prod\nolimits_{k=t+1}^s M^{(k)}\Big)^{\top}.
\end{align}}%
The derivation of statement \eqref{eq_tv_ec} is now completed. \hfill$\square$

\section{Epidemic Models}\label{sec_epi_models}

We describe the epidemic models that are considered in our experiments. 
One widely used model of epidemic spread is the SEIR compartmental model.
An SEIR model uses four compartments to capture a spreading process: Susceptible (S), Exposed (E), Infected (I), and Recovered (R).
Every node must belong to one of the four states during the process.
At every time $t$,
\begin{itemize}[leftmargin=1cm]
    \setlength\itemsep{0.00em}
	\item $S^{(t)}$ denotes the set of susceptible nodes at time $t$. A node may get exposed if its incoming neighbors are infectious. The probability depends on the edge weights and the virus transmission rate.
	\item $E^{(t)}$ denotes the nodes exposed to the virus but are not infectious at time $t$. In expectation, a node remains exposed for $\delta_E$ periods.
	\item $I^{(t)}$ denotes the nodes who are infectious at time $t$. Each node remains infectious for $\delta_I$ periods in expectation.
	\item $R^{(t)}$ denotes the nodes who have recovered at time $t$.
\end{itemize}

For weighted graphs, we simulate an SEIR model. At each time $t$, we calculate the infection probability for node $i$ based on the edge weights and transmission rate $\beta_{\text{Base}}$:
{\small
$$
p_i = 1 - \prod_{(i,j)\in E: j \in I^{(t)}} \max\Big(1 - W_{i,j}\beta_{\text{base}}, 0\Big).
$$
}

We follow the procedure in \citet{chang2021mobility} for mobility networks to simulate the metapopulation SEIR model.
At time $t$, the transitions between the four states (for $c_i$) are sampled as follows:
\begin{small}
\begin{align}
	N^{(t)}_{S_{c_i} \rightarrow E_{c_i}} &\sim \pois\Big( \frac{S_{c_i}^{(t)}}{N_{c_i}} \lambda^{(t)} \Big) + \Binom\Big(S_{c_i}^{(t)}, \lambda_{c_i}^{(t)}\Big). \label{eq_S_E} \\
	N_{E_{c_i} \rightarrow I_{c_i}}^{(t)} &\sim \Binom\Big(E_{c_i}^{(t)}, \frac{1}{\delta_E}\Big). \label{eq_E_I} \\
	N_{I_{c_i} \rightarrow R_{c_i}}^{(t)} &\sim \Binom\Big(I_{c_i}^{(t)}, \frac{1}{\delta_I}\Big). \label{eq_I_R}
\end{align}
\end{small}

\noindent where $\lambda^{(t)}$ is the aggregate transmission rate over the points of interest; $\lambda_{c_i}^{(t)}$ is the base transmission rate within $c_i$; $\latent$ represents the mean latency period; $\recover$ is the mean infectious period.

In equation \eqref{eq_S_E},
$\lambda_{c_i}^{(t)}$ is given by the product of the base transmission rate $\beta_{\base}$ and the proportion of infectious individuals in CGB $c_i$: 
    $\lambda_{c_i}^{(t)} = \beta_{\base} \frac{I_{c_i}^{(t)}}{N_{c_i}}.$
The infection rate across all the POIs is 
{\small
$$
    \lambda^{(t)} = \sum_{j=1}^n \lambda_{p_j}^{(t)} W_{i,j}^{(t)}; \ \lambda_{p_j}^{(t)} = \beta_{p_j}^{(t)} \frac{ I_{p_j}^{(t)}}{\sum_{i=1}^m W_{i, j}^{(t)} }.
$$}

\noindent where $\lambda_{p_j}^{(t)}$ is the infection rate for POI $p_j$  at time $t$.
$\beta_{p_j}^{(t)}$ is the transmission rate at POI $p_j$ and ${I_{p_j}^{(t)}}$ is the number of infectious individuals in $p_j$ at time $t$. The parameters are estimated as follows.
(i) $\beta_{p_j}^{(t)}$ is estimated by the physical area of $p_j$: 
$
\beta_{p_j}^{(t)} = \psi \cdot \cdot d_{p_j}^2 \cdot \frac{V_{p_j}^{(t)}}{a_{p_j}}$
in which $\psi$ is a transmission constant;
$a_{p_j}$ is the physical area of $p_j$;
$V_{p_j}^{(t)} = \sum_{i=1}^m W_{i,j}^{(t)}$ represents the number of visitors to $p_j$ at time $t$.
(ii) $I_{p_j}^{(t)}$ is estimated in proportion to the infectious population from each CBG and their number of visits to $p_j$: 
$
I_{p_j}^{(t)} = \sum_{k=1}^m \frac{I_{c_k}^{(t)}}{N_{c_k}} W_{k,j}^{(t)}.
$

The SEIR model has many variants (cf. \citet{prakash2012threshold}). We consider SIR and SIS that share similar spreading processes as the SEIR model. We describe their differences as follows.
The SIR model uses three compartments as the SEIR model except for the exposed state. It assumes that there is no latent period of the disease. Nodes are capable of infecting susceptible nodes directly after being infected.
The SIS model uses two states (Susceptible and Infectious) in a spreading process. It assumes that recovery does not bring immunity and nodes who have recovered will become susceptible again. 

\section{Experiment Details}\label{sec_add_setup}

\medskip
\noindent\textbf{Simulation setup.} 
For the weighted graphs, we simulate an SEIR model on each graph.
We use a transmission rate $\beta_{\text{Base}} = 0.05$ and a initial exposed ratio $p_0 = 0.01$.
To avoid infecting all the graph nodes, we simulate for 50 epochs. We use a slightly higher edge-weight reduction budget as 20\% of the total edge weights because the average edge weight in these three graphs is smaller than the mobility networks.

For the experiments concerning mobility networks, we follow the procedures of \citet{chang2021mobility} to simulate a metapopulation SEIR model in each network.
We calibrate the parameters of the SEIR model following their method.
We simulate 100 epochs on static mobility networks to be consistent with the simulation of \citet{chang2021mobility}.
The results are consistent throughout the simulation.
We compare the \LPshort{} algorithm with baseline methods using an edge-weight reduction budget as 5\% of the total edge weights. The results of using other budget amounts are consistent. We use the same set of parameters for SIR and SIS model simulations.

In time-varying mobility networks experiments, we simulate the metapopulation SEIR model on a sequence of ten networks for 70 epochs for every network. We set the edge-weight reduction budget as 5\% of the total edge weights of the sequence.

\medskip
\noindent\textbf{Model validation.} 
We calibrate the following parameters for the metapopulation SEIR model on mobility networks: (i) the transmission constant in POIs, $\psi$; (ii) the base transmission rate, $\beta_{\text{base}}$; and (iii) the ratio of initially exposed individuals, $p_0$.
We use grid search to find the parameters with the smallest root mean square error compared to the reported number of infected cases.  
We calibrate an SEIR model for every MSA independently.
We compare the predicted cases of our simulated SEIR model with the reported cases from New York Times COVID-19 data.
The root mean squared error of all the epochs is 295.17, averaged over eight mobility networks. The error is within $3\%$ compared to the overall infected population at $10^4$. These results reaffirm the finding of \citet{chang2021mobility}.

\medskip
\noindent\textbf{Data availability.}
The three weighted graphs are available in the following sources: Airport\footnote{\url{http://opsahl.co.uk/tnet/datasets/openflights.txt}}, Adavogato\footnote{\url{https://downloads.skewed.de/mirror/konect.cc/files/download.tsv.advogato.tar.bz2}}, and Bitcoin\footnote{\url{http://snap.stanford.edu/data/soc-sign-bitcoinalpha.html}}.
The two weighted time-varying graphs are available in the following sources: Bitcoin-Alpha\footnote{\url{https://snap.stanford.edu/data/soc-sign-bitcoinalpha.csv.gz}} and Bitcoin-OTC\footnote{\url{https://snap.stanford.edu/data/soc-sign-bitcoinotc.csv.gz}}. We report the network statistics in Table \ref{table_web}. 
The mobility network data is freely available to researchers, non-profit organizations, and governments through the SafeGraph COVID-19 Data Consortium.\footnote{\url{https://www.safegraph.com/covid-19-data-consortium}}
The construction of mobility networks requires the following data sources:
(i) Mobility patterns from the Monthly Pattern\footnote{\url{https://docs.safegraph.com/docs/monthly-patterns}} and Weekly Pattern datasets, \footnote{\url{https://docs.safegraph.com/docs/weekly-patterns}}
(ii) The geometry dataset,\footnote{\url{https://docs.safegraph.com/docs/geometry-data}},
(iii) The Open Census Dataset\footnote{\url{https://docs.safegraph.com/docs/open-census-data}}, and
(iv) The New York Times COVID-19 data.\footnote{\url{https://github.com/nytimes/covid-19-data}} 

\begin{table}[t!]
\caption{\textbf{Left:} Dataset statistics for three weighted graphs. \textbf{Right:} Dataset statistics for four time-varying networks. Each time-varying network sequence has ten networks.}\label{table_web}
\begin{minipage}[b]{0.43\textwidth}
\centering
{\footnotesize\begin{tabular}{@{}lccc@{}}
        \toprule
                & Airport & Advogato & Bitcoin       \\ \midrule
            Nodes     & 7,977 & 6,541     & 3,783                 \\
            Edges     & 30,501 & 51,127    & 24,186    \\
            Avg. edge weight & 1.45 & 0.83 & 1.46      \\ \bottomrule
\end{tabular}}
\end{minipage}
\begin{minipage}[b]{0.55\textwidth}
\centering
{\footnotesize\begin{tabular}{@{}lcccc@{}}
        \toprule
                 & Bitcoin-Alpha & Bitcoin-OTC & Chicago & Houston  \\\midrule
Nodes            & 3,783 & 5,881 & 32,390 & 38,895                \\
Edges            & 24,186 & 35,591 & 975,569  & 1,586,683  \\
Avg. edge weight & 1.46 & 1.51 & 4.27 & 4.42   \\ \bottomrule
\end{tabular}}
\end{minipage}
\end{table}

\end{document}